\documentclass{article}
\usepackage{amsmath,inputenc}

\title{Online Learning for Equilibrium Pricing in Markets under Incomplete Information}
\author{Devansh Jalota$^{*\dagger}$, Haoyuan Sun$^{*\ddagger}$, and Navid Azizan$^\ddagger$\\
\thanks{$^*$Equal Contribution}
\\
$^\dagger$Stanford University\\
$^\ddagger$Massachusetts Institute of Technology
}
\date{}

\usepackage{graphicx}
\usepackage{csvsimple}
\usepackage{rotating}
\usepackage[letterpaper, portrait, margin=1in]{geometry}
\usepackage{hyperref}
\hypersetup{
    colorlinks=true,
    linkcolor=blue,
    citecolor=red,
    filecolor=magenta,      
    urlcolor=cyan
    }

\usepackage[nocomma]{optidef}
\usepackage{listings}
\usepackage{comment}
\usepackage{appendix}
\usepackage[ruled,vlined]{algorithm2e}
\usepackage{amsfonts} 
\usepackage{amssymb}
\usepackage{bm}
\usepackage[square,comma,numbers,sort]{natbib}

\usepackage{tikz}
\usepackage{pgfplots}
\DeclareUnicodeCharacter{2212}{−}
\usepgfplotslibrary{groupplots,dateplot}
\usetikzlibrary{patterns,shapes.arrows}
\pgfplotsset{compat=1.17}

\usepackage{caption}
\usepackage{subcaption}

\usepackage{color} \definecolor{mygreen}{RGB}{28,172,0} \definecolor{mylilas}{RGB}{170,55,241}
\DeclareFixedFont{\ttb}{T1}{txtt}{bx}{n}{12} \DeclareFixedFont{\ttm}{T1}{txtt}{m}{n}{12}  \usepackage{bbm}
\usepackage{amsmath,amsthm}

\newtheorem{theorem}{Theorem}
\newtheorem{corollary}{Corollary}
\newtheorem{proposition}{Proposition}

\newtheorem{lemma}{Lemma}

\newtheorem*{claim}{Claim}

\theoremstyle{definition}
\newtheorem{definition}{Definition}
\newtheorem{remark}{Remark}
\newtheorem{example}{Example}

\newcommand{\norm}[1]{\left\lVert#1\right\rVert}

\usepackage{color}
\definecolor{deepblue}{rgb}{0,0,0.5}
\definecolor{deepred}{rgb}{0.6,0,0}
\definecolor{deepgreen}{rgb}{0,0.5,0}

\usepackage{listings}

\renewcommand\footnotemark{}

\def\S{\mathcal{S}}
\def\ppi{\boldsymbol{\pi}}

\newcommand{\RR}{\mathbb{R}}
\newcommand{\EE}{\mathbb{E}}

\newcommand{\est}{\textsc{Est}_{\mathrm{sq}}}
\newcommand{\preg}{\widetilde{\textsc{Reg}}}

\DeclareMathOperator*{\argmin}{\mathrm{arg\,min}}
\DeclareMathOperator*{\argmax}{\mathrm{arg\,max}} 

\begin{document}

\maketitle

\begin{abstract}
The computation of \emph{equilibrium} prices at which the supply of goods matches their demand typically relies on complete information on agents’ private attributes, e.g., suppliers' cost functions, which are often unavailable in practice. Motivated by this practical consideration, we consider the problem of learning equilibrium prices over a horizon of $T$ periods in the incomplete information setting wherein a market operator seeks to satisfy the customer demand for a commodity by purchasing it from competing suppliers with cost functions unknown to the operator. We first consider the setting when suppliers' cost functions are fixed and develop algorithms that, on three pertinent regret metrics, simultaneously achieve a regret of \ $O(1)$ \ when the customer demand is constant over time, and \ $O(\sqrt{T})$ \ when the demand varies over time. In the setting when the suppliers' cost functions vary over time, we demonstrate that, in general, no online algorithm can achieve sublinear regret on all three metrics. Thus, we consider an augmented setting wherein the operator has access to hints/contexts that reflect the variation in the cost functions and propose an algorithm with sublinear regret in this augmented setting. \ Finally, we present numerical experiments that validate our results and discuss various model extensions.  \

\end{abstract}

\section{Introduction}

The study of market mechanisms for efficiently allocating scarce resources traces back to the seminal work of~\citet{walras1954elements}. In his work, Walras investigated the design of pricing schemes to mediate the allocation of scarce resources such that the economy operates at an \emph{equilibrium}, i.e., the supply of each good matches its demand. Market equilibria exist under mild conditions on agents' preferences~\citep{arrow-debreu} and, under convexity assumptions on their preferences, can often be computed by solving a large-scale centralized optimization problem. As a case in point, in electricity markets with convex supplier cost functions, the equilibrium prices correspond to the shadow prices~\citep{beato,turvey} of a convex optimization problem that minimizes the sum of the supplier costs subject to a market-clearing (or load-balance) constraint~\citep{azizan2020optimal}.

While methods such as convex programming provide computationally tractable approaches to computing market equilibria, the efficacy of such centralized optimization approaches for equilibrium computation suffers from several inherent limitations. First, centralized optimization approaches rely on complete information on agents' utilities and cost functions that are often unavailable to a market operator. For instance, with the deregulation of electricity markets, suppliers' cost functions are private information, which has led to strategic bidding practices by suppliers seeking to maximize their profits~\citep{azizan2017opportunities,liberopoulos2016critical,VENTOSA2005897,david2000strategic} and has been associated with tens of millions of dollars of over-payments to suppliers~\citep{strategic-bidding-jp-morgan}. Moreover, even if a market operator has access to some information on agents' utilities and cost functions, such information can typically only provide a noisy or imperfect estimate of their preferences due to inadequate information or uncertainty~\citep{weitzman-pvq}. In the context of electricity markets, the advent of renewables and distributed energy resources has accompanied a high degree of uncertainty in the supply of energy to meet customer demands at different times of the day and year, as these energy sources are sensitive to weather conditions. To further compound these challenges, agents' preferences in markets such as electricity markets may also be time-varying, e.g., in electricity markets, customer demands may change over time and the cost functions of suppliers may depend on fluctuating weather conditions. Thus, a market operator may need to periodically collect agents' preferences and solve a large-scale centralized optimization at each time period to set equilibrium prices, which may be computationally challenging.

Motivated by these practical considerations that limit the applicability of centralized optimization approaches to computing equilibrium prices, in this work, we study the problem of setting equilibrium prices in the incomplete information setting where a market operator seeks to satisfy customer demand for a commodity by purchasing the required amount from competing suppliers with privately known cost functions. We investigate this problem under several informational settings regarding the time-varying nature of the customer demands and supplier cost functions and develop online learning \emph{posted-price} algorithms for each of these settings that iteratively adjust the prices in the market over time. Our proposed algorithms employ the observation that a market operator can effectively learn information on suppliers' costs and equilibrium prices through observations of their cumulative production relative to the customer demand given different market prices. To analyze the performance of our algorithms, we combine techniques from online learning and parametric optimization as we seek to simultaneously optimize multiple, often competing, performance metrics pertinent in the context of equilibrium pricing.

\subsection{Contributions}

In this work, we study the problem of setting equilibrium prices faced by a market operator that seeks to satisfy an inelastic customer demand for a commodity by purchasing the required amount from $n$ competing suppliers. Crucially, we study this problem in the incomplete information setting when the cost functions of suppliers are private information and thus unknown to the market operator. Since traditional centralized methods of setting equilibrium prices are typically not conducive in this incomplete information setting, we consider the problem of learning equilibrium prices over $T$ periods to minimize three performance (regret) metrics: (i) \emph{unmet demand}, (ii) \emph{cost regret}, and (iii) \emph{payment regret}. Here unmet demand refers to the cumulative difference between the demand and the total production of the commodity corresponding to an online pricing policy. Furthermore, cost regret (payment regret) refers to the difference between the total cost of all suppliers (payment made to all suppliers) corresponding to the online allocation and that of the offline oracle with complete information on suppliers' cost functions. For a more thorough discussion of these regret metrics, we refer to Section~\ref{sec:perf-measures}.

In this incomplete information setting, we investigate the design of posted-price online algorithms to set a sequence of prices that achieve sub-linear regret, in the number of periods $T$, on the above three performance metrics. To this end, we first consider the setting when suppliers' cost functions are fixed over the $T$ periods and develop algorithms that achieve a regret of \ $O(1)$ \ when the customer demand is constant over time (Section~\ref{sec:fixed-setting}), and \ $O(\sqrt{T})$ \ when the demand is variable over time (Section~\ref{sec:vary-demand}), for strongly convex cost functions.  Notably, our algorithms do not rely on any information on the time horizon $T$ to achieve these regret guarantees. To establish these regret guarantees for the three performance metrics, we leverage and combine techniques from parametric optimization and online learning. We further demonstrate through an example that if the strong convexity condition
on suppliers' cost functions is relaxed, then, in general, no online algorithm can achieve a sub-linear regret guarantee on all three regret metrics.

Then, we consider the setting when suppliers' cost functions can vary across the $T$ periods (Section~\ref{sec:time-vary-cost-main}) and show that if the operator does not know the process that governs the variation of the cost functions, no online algorithm can achieve sub-linear regret on all three regret metrics. Thus, in alignment with real-world markets, e.g., electricity markets, we consider an augmented setting, wherein the market operator has access to some hints (contexts) that, without revealing the complete specification of the cost functions, reflect the change of the cost functions over time. In this setting, we propose a posted-price algorithm that achieves sub-linear regret on all three performance metrics, where the exact dependence of the regret guarantee on $T$ relies on the statistical properties of the function class that suppliers' cost functions belong to.

\
Next, we present numerical experiments that further corroborate our regret guarantees (Section~\ref{sec:experiments}) based on two market instances inspired
by the electricity markets literature. In particular, we validate the regret bounds for Algorithms~\ref{alg:fixed-demand-constant-reg} and~\ref{alg:time-varying-demand-new-sqrt} on these two market instances for the settings with fixed costs and demands and fixed cost but variable demands, respectively. 
\
We further validate our impossibility result on developing an algorithm with sub-linear regret in the variable cost function setting (without access to contexts) by numerically showing that a widely studied primal-dual algorithm in the online learning literature achieves linear regret on at least one of the three performance metrics in our studied problem setting.

\
Finally, in Section~\ref{sec:discussion}, we outline several natural extensions of the model studied in this work that incorporate key practical considerations and offer promising directions for further research. To this end, we first present an illustrative example that highlights the challenges in considering forward-looking suppliers who strategically misreport production in certain periods to induce more favorable prices in subsequent periods. Such behavior can lead to incorrect learning of equilibrium prices, underscoring the need for algorithms robust to strategic manipulation. We then explore relaxations of the unmet demand regret metric, e.g., allowing for storage, and demonstrate that even under some of these relaxed formulations, the impossibility results established in this work continue to hold.
\

\section{Literature Review}

The design of market mechanisms to efficiently allocate resources under incomplete information on agents' preferences and costs has received considerable attention in the operations research, economics, and computer science communities. For instance, mechanism design has enabled designing optimal resource allocation strategies even in settings when certain information is privately known to agents~\citep{akbarpour2020redistributive,NIPS2004_fc03d482,heydaribeni2018distributed,pmlr-v119-deng20d}. Furthermore, inverse game theory~\citep{data-inverse-opt} and revealed preference based approaches~\citep{balcan2014learning,bei2016learning,beigman2006learning,zadimoghaddam2012efficiently} have emerged as methods to learn the underlying utilities and costs of agents given past observations of their actions. While, in line with these works, we consider an incomplete information setting wherein suppliers' cost functions are private information, we do not directly learn or elicit suppliers' cost functions to make pricing decisions as in these works and instead study the problem of learning equilibrium prices as an online decision-making problem.

The paradigm of online decision making has enabled the allocation of scarce resources in settings with incomplete information where data is revealed sequentially to an algorithm designer and has found several applications~\citep{msvv,manshadi2021fair,lien2014sequential}. Two of the most well-studied classes of online decision-making problems include online linear programming (OLP) and online convex optimization (OCO). While OLP has been studied extensively under different informational settings, including the adversarial~\citep{msvv,TCS-057}, random permutation~\citep{online-agrawal,devanur-adwords}, and stochastic input models~\citep{li2020simple,li2021symmetry,chen2021linear,li2021online}, in this work, we consider the setting when suppliers' cost functions are convex and, in general, non-linear. Given the prevalence of non-linear objectives in various resource allocation settings, there has been a growing interest in OCO~\citep{OPT-013}, wherein several works~\citep{AgrawalD15,balseiro2022best,balseiro2021regularized} have investigated the design of algorithms with near-optimal regret guarantees under the adversarial, random permutation, and stochastic input models. 
Additionally, there have been many works on a smoothed variant of OCO where the agent must pay switching costs for changing decisions between rounds~\citep{lin2012online,goel2019beyond,chen2018smoothed,bansal20152}.
As in the works on OCO, we also consider general non-linear convex objectives; however, as opposed to the resource constraints that only need to be satisfied in aggregate over the entire time horizon in these works, we adopt a stronger performance metric where we accumulate regret at each period when the customer demand is not satisfied (see Section~\ref{sec:perf-measures} for more details on our performance metrics).
Thus, our algorithms are considerably different from the algorithms based on applying dual sub-gradient descent developed in these works on OCO~\citep{AgrawalD15,balseiro2022best,balseiro2021regularized}.

Our algorithms are inspired by the multi-armed bandit literature and involve a tradeoff between exploration and exploitation in an unknown environment~\citep{freund1999adaptive,auer2002nonstochastic}.
In a typical multi-armed bandit (MAB) setting, a decision-making agent performs sequential trials on a set of permissible actions (arms), observes the outcome of the actions, and maximizes its rewards.
Several extensions of MAB have been proposed over the years, including bandits with partial observations~\citep{bartok2014partial}, contextual bandits~\citep{langford2007epoch,slivkins2011contextual}, Lipschitz bandits~\citep{mab-ms-2008}, and bandits with constrained resources~\citep{agrawal2016efficient, badanidiyuru2018bandits}.
These results have contributed to many applications, such as online posted-price auctions~\citep{oppa}, dynamic pricing with limited supply~\citep{babaioff2015dynamic}, and dark pools in the stock market~\citep{agarwal2010optimal}.
In typical bandit frameworks, the decision maker's objective is to optimize a single reward function, wherein the rewards are revealed sequentially as part of the observation to the decision maker.
However, in our setting, we seek to jointly optimize multiple performance metrics where suppliers' cost functions are not revealed to the market operator (see Section~\ref{sec:perf-measures} for details).
 
\section{Model}
\label{sec:model}

In this section, we present the offline model of a market operator seeking to set equilibrium prices in the market to satisfy a customer demand for a commodity (Section~\ref{sec:market-model}) and performance metrics used to evaluate the efficacy of a pricing policy in the online setting (Section~\ref{sec:perf-measures}).

\subsection{Market Model and Equilibrium Pricing} \label{sec:market-model}

We study a market run by an operator seeking to meet the customer demand $d>0$ for a commodity, e.g., energy, by purchasing the required amount from $n$ competing suppliers. Each supplier $i \in [n]$ has a cost function $c_i: \mathbb{R}_{\geq 0} \rightarrow \mathbb{R}_{\geq 0}$, where $c_i(x_i)$ represents the cost incurred by supplier $i$ for producing $x_i$ units of the commodity. Furthermore, to meet the customer demand, the market operator posts a price $p$ in the market, which represents the payment made by the market operator for each unit of the commodity produced by a given supplier. In particular, for producing $x_i$ units of the commodity, a supplier $i$ receives a payment of $p x_i$ from the market operator. Then, given a posted price $p$ for the commodity and a cost function $c_i(\cdot)$, each supplier makes an individual decision on the optimal production quantity $x_i^*(p)$ to maximize their total profit, as described through the following optimization problem
\begin{maxi}|s|[2]                   {x_i \geq 0}                               {p x_i - c_i(x_i). \label{eq:supObj}}   {}             {}                                \end{maxi}

The posted price that suppliers best respond to is set by a market operator that seeks to determine an equilibrium price $p^*$ that satisfies the following three desirable properties: 

\begin{enumerate}
    \item Market Clearing: The total supply equals the total demand, i.e., $\sum_{i = 1}^n x_i^*(p^*) = d$.
    \item Minimal Supplier Production Cost: The total production cost of all suppliers, given by $\sum_{i = 1}^n c_i(x_i^*(p^*))$, is minimal among all feasible production quantities $x_i \geq 0$ for all suppliers $i \in [n]$ satisfying the customer demand, i.e., $\sum_{i = 1}^n x_i = d$.
    \item Minimal Payment: The total payment made to all suppliers, given by $\sum_{i = 1}^n p^* x_i^*(p^*)$, is minimal among all feasible production quantities $x_i \geq 0$ for all suppliers $i \in [n]$ satisfying the customer demand. \end{enumerate}

While these properties are, in general, not possible to achieve simultaneously, e.g., in markets where the supplier cost functions are non-convex~\citep{azizan2020optimal}, in markets where the cost functions $c_i(\cdot)$ of all suppliers are convex, there exists an equilibrium price $p^*$ that satisfies the above three properties. Moreover, in markets with convex cost functions, the equilibrium price can be computed through the dual variables of the market-clearing constraint of the following convex optimization problem
\begin{mini!}|s|[2]                   {x_{i} \geq 0, \forall i \in [n]}                               {\sum_{i = 1}^n c_i(x_{i}), \label{eq:supObj2}}   {\label{eq:minCost}}             {C^* = }                                \addConstraint{\sum_{i = 1}^n x_{i}}{= d, \label{eq:demand-con}} 
\end{mini!}
where~\eqref{eq:supObj2} is the minimum supplier production cost objective and~\eqref{eq:demand-con} is the market clearing constraint. 
In particular, from the KKT condition, the optimal solution to Problem~\eqref{eq:supObj2}-\eqref{eq:demand-con} satisfies
\[
\begin{cases}
    \sum_{i=1}^n x^*_i = d, x_i^* \geq 0, \forall \,  i = 1, \dots, n, \\
    \frac{\partial c_i}{\partial x_i}(x^*_i) \ge p^*, \forall \,  i = 1, \dots, n, \\
    \frac{\partial c_i}{\partial x_i}(x^*_i) = p^*, \forall \,  i \text{ s.t. } x_i^*>0,
\end{cases}
\]
so that the optimal dual variable $p^*$ satisfies all conditions of equilibrium pricing.
While the equilibrium price $p^*$ has several desirable properties, such an equilibrium price typically cannot be directly computed by solving Problem~\eqref{eq:supObj2}-\eqref{eq:demand-con} as the cost functions of suppliers are, in general, unknown to the market operator. Furthermore, both the cost functions of the suppliers and the customer demands tend to be time-varying and thus would involve the market operator periodically re-solving Problem~\eqref{eq:supObj2}-\eqref{eq:demand-con} to determine equilibrium prices at short time intervals, which may be computationally prohibitive. To overcome these challenges, in this work, we propose online learning algorithms to learn equilibrium prices over multiple periods in the incomplete information setting when the cost functions of the suppliers are unknown (or only partially known) to the market operator.

\subsection{Performance Metrics to Set Equilibrium Prices in Online Setting} \label{sec:perf-measures}

We now introduce the online learning setting, wherein the market operator sets prices for the commodity over multiple periods, and present the performance metrics to evaluate the efficacy of an online pricing policy. In particular, we consider the setting when the market operator seeks to satisfy the customer demand over multiple periods $t = 1, \ldots, T$. At each period $t \in [T]$, the customer demand for the commodity is given by $d_t$ and each supplier $i \in [n]$ has a private cost function $c_{it}(\cdot)$ that is increasing, continuously differentiable, $\mu$-strongly convex, $\ell$-Lipschitz smooth, and normalized to satisfy $c_{it}(0) = 0$. 
\
We note that these assumptions on supplier cost functions are indeed standard in the study of electricity markets, and strongly convex cost functions, e.g., quadratic or smooth piece-wise quadratic cost functions, often serve as analytically tractable approximations to cost functions observed in practice (see, e.g.,~\cite{5628271,TSAOUSOGLOU2022111890,1102918,7796713,4112652}). 
\

We assume that the demand at each period $t$ lies in a bounded interval, i.e., $d_t \in [\underline{d}, \Bar{d}]$ for all $t$ for some $\underline{d}, \Bar{d}>0$. Furthermore, for ease of exposition, we normalize the set of feasible prices corresponding to any customer demand and realization of supplier cost functions to be such that the corresponding optimal price of the commodity belongs to the normalized interval $[0, 1]$.
\
We highlight that we only consider strongly convex cost functions of suppliers, as opposed to general convex costs. As we shall explain in the next section, in general, there is a fundamental performance limitation of any online algorithm in the incomplete information setting studied in this work for non-strongly convex cost functions.
\

In this work, we begin by considering the informational setting wherein the cost functions of the suppliers are fixed over time (see Sections~\ref{sec:fixed-setting} and~\ref{sec:vary-demand}) and upon observing the customer demand $d_t$, the market operator makes a pricing decision  (i.e., posts a price) $p_t$ that depends on the past observations of supplier productions, i.e., revealed preference feedback in response to set prices as in~\citet{roth2016watch,ji2018social,pmlr-v206-jalota23a,ofm-2023}, and the realized customer demands. In particular, over the $T$ periods, the market operator sets a sequence of prices given by the pricing policy $\ppi = (\pi_1, \ldots, \pi_T)$, where $p_t = \pi_t(\{ (x_{it'}^*)_{1 = 1}^n, d_{t'} \}_{t'=1}^{t-1}, d_t)$, where $x_{it}^*$ represents the optimal production quantity corresponding to the solution of Problem~\eqref{eq:supObj} for supplier $i$ at period $t$. 
When the pricing policy is evident from the context, we will overload the notation and simply write $\ppi = (p_1, p_2, \dots, p_T)$.
We then consider the informational setting when suppliers' cost functions are time-varying (see Section~\ref{sec:time-vary-cost-main}), for which we introduce an augmented problem setting and the corresponding class of online pricing policies that we consider in Section~\ref{sec:contextual-bandit}.

We evaluate the efficacy of an online pricing policy $\ppi$ using three regret metrics: (i) \textit{unmet demand}, (ii) \textit{cost regret}, and (iii) \textit{payment regret}. These regret metrics represent the performance loss of the policy $\ppi$ relative to the optimal offline algorithm with complete information on the three desirable properties of equilibrium prices elucidated in Section~\ref{sec:market-model}. We also note that these performance metrics naturally generalize to the augmented problem setting we consider when suppliers' cost functions are time-varying and present the corresponding generalizations of the regret metrics in Appendix~\ref{sec:new-regret-defs} for completeness.

\paragraph{Unmet Demand:} We evaluate the unmet demand of an online pricing policy $\ppi$ as the sum of the differences between the demand and the total supplier productions corresponding to the pricing policy $\ppi$ at each period $t$. In particular, for an online pricing policy $\ppi$ that sets a sequence of prices $p_1, \ldots, p_T$, the cumulative unmet demand is given by \vspace{-5pt}
\begin{align*}
    U_T(\ppi) = \sum_{t = 1}^T \left( d_t - \sum_{i = 1}^n x_{it}^*(p_t) \right)_+,
\end{align*}
where $x_{it}^*(p_t)$ is the optimal solution of Problem~\eqref{eq:supObj} for supplier $i$ at period $t$. 

\paragraph{Cost Regret:} We evaluate the cost regret of an online pricing policy $\ppi$ through the difference between the total supplier production cost corresponding to algorithm $\ppi$ and the minimum total production cost, given complete information on the supplier cost functions. Specifically, the cost regret $C_T(\ppi)$ of a policy $\ppi$ is \vspace{-5pt}
\begin{align*}
    C_T(\ppi) = \sum_{t = 1}^T \sum_{i = 1}^n \left(c_{it}(x_{it}^*(p_t)) - c_{it}(x_{it}^*(p_t^*))\right),
\end{align*}
where the price $p_t^*$ for each period $t \in [T]$ is the optimal dual variable of the market clearing constraint of Problem~\eqref{eq:supObj2}-\eqref{eq:demand-con} given the demand $d_t$ and cost functions $c_{it}$ for all $i \in [n]$.

\paragraph{Payment Regret:} Finally, we evaluate the payment regret of online pricing policy $\ppi$ through the difference between the total payment made to all suppliers corresponding to algorithm $\ppi$ and the minimum total payment, given complete information on the supplier cost functions. In particular, the payment regret $P_T(\ppi)$ of an algorithm $\ppi$ is given by \vspace{-5pt}
\begin{align*}
    P_T(\ppi) = \sum_{t = 1}^T \sum_{i = 1}^n \left(p_t x_{it}^*(p_t) -  p^*_t x_{it}^*(p^*_t) \right).
\end{align*}

In this work, we focus on developing algorithms that jointly optimize these three regret metrics over $T$ periods. 
In particular, because it is desirable that the performance of the pricing policy improves as the market operator receives more information, we intend to design algorithms that guarantee all three regret metrics are sub-linear in $T$.

A few comments about our regret metrics are in order. First, our unmet demand metric aligns with real-world markets, e.g., electricity markets, where the demand needs to be satisfied at each period, and overproduction at particular periods cannot compensate for unmet demand at subsequent periods. 
Therefore, we define our unmet demand metric as a stronger benchmark than the typical constraint violation metrics in the literature of jointly optimizing multiple regret metrics~\citep{yu2017online,jenatton2016adaptive,mahdavi2012trading}, where resource constraints only need to be approximately satisfied in the aggregate. 
Formally, $U_T(\ppi) = \sum_{t = 1}^T \left( d_t - \sum_{i = 1}^n x_{it}^*(p_t) \right)_+ \geq \left[ \sum_{t = 1}^T \left( d_t - \sum_{i = 1}^n x_{it}^*(p_t) \right) \right]_+$, where the latter term corresponds to the setting when the customer demand only needs to be satisfied in the long-run. Further, since we obtain regret guarantees for the above unmet demand metric, using techniques from parametric optimization, our regret guarantees naturally extend for the corresponding stronger notions of the payment and cost regret metrics as well. However, we present our payment and cost regret metrics in alignment with the classical regret metrics in the literature, wherein lower payments (costs) at particular periods can compensate for excess payments (costs) at other periods. 

\section{Fixed Cost Functions and Demand} \label{sec:fixed-setting}

We now investigate the design of online pricing policies achieving good performance on the three regret metrics, i.e., sub-linear unmet demand, cost regret, and payment regret in the number of periods $T$. As a warm-up, we first consider the setting when the cost functions of the suppliers and customer demand are fixed over the $T$ periods. 
Formally, the supplier cost functions satisfy $c_{it}(\cdot) = c_{it'}(\cdot)$ for all $t, t' \in [T]$ and the demands satisfy $d_t = d_{t'}$ for all periods $t, t' \in [T]$. For ease of exposition, in this section, we drop the subscript $t$ in the notation for the customer demand (and supplier cost functions) and denote $d_t = d$ (and $c_{it}(\cdot) = c_i(\cdot)$ for all suppliers $i$) for all periods $t \in [T]$.
In this setting, we develop an algorithm that achieves a regret of \ $O(1)$ \ on the three regret metrics when the suppliers' cost functions are strongly convex (Sections~\ref{sec:fixed-algorithm} and~\ref{sec:proof-prop1}). We further present an example to demonstrate that if the strong convexity condition on suppliers' cost functions is relaxed, then no sub-linear regret guarantee on all three regret metrics is, in general, possible for any online algorithm (Section~\ref{sec:fixed-convex-limitations}).

\subsection{Algorithm with Sub-linear Regret for Strongly Convex Cost Functions}
\label{sec:fixed-algorithm}

In this section, we consider the setting of fixed supplier cost functions and customer demands and present an algorithm \ achieving constant regret \ on the three performance metrics when the suppliers' cost functions are strongly convex. 
To motivate our algorithm, we first note that since the customer demand and the supplier cost functions are fixed over time, the optimal price $p^* \in [0, 1]$ for all periods $t \in [T]$ is also fixed and given by the dual of the demand constraint of Problem~\eqref{eq:supObj2}-\eqref{eq:demand-con}. 
Furthermore, the cumulative production $x_t^*(p) = \sum_{i = 1}^n x_{it}^*(p)$ is monotonically non-decreasing in the price $p$ because suppliers' cost functions are increasing. Utilizing this monotonicity property, we note that if we set two prices $p_1, p_2 \in [0, 1]$ such that the cumulative production $\sum_{i = 1}^n x_{it}^*(p_1)>d$ and $\sum_{i = 1}^n x_{it}^*(p_2)<d$, then $p_1$ and $p_2$ respectively serve as upper and lower bounds on the optimal price $p^*$ when the supplier cost functions and customer demands are fixed over time.

\
Following these observations on the monotonicity of the cumulative production in the prices, we present Algorithm~\ref{alg:fixed-demand-constant-reg}, akin to performing binary search, which maintains a feasible interval for the optimal price $p^*$ and sets a price at the middle of that interval for each arriving user to continuously shrink this feasible price set. In particular, given a feasible price interval $[a, b]$, the algorithm offers a price $\frac{a+b}{2}$, and if the total supply exceeds (is less than) the demand at the offered price, then $[a, (a+b)/2]$ (resp. $[(a+b)/2, b]$) is set as the new feasible interval.
This process of shrinking the feasible interval is presented formally in Algorithm~\ref{alg:fixed-demand-constant-reg}.

\begin{algorithm} 
\SetAlgoLined
\SetKwInOut{Input}{Input}\SetKwInOut{Output}{Output}
Initialize feasible set of prices $\S_p := [a_1, b_1] = [0, 1]$\;
Set $p_t \gets \frac{a_1+b_1}{2} = \frac{1}{2}$ \;
 \For{$t = 1, 2, \dots$}{
 Offer price $p_t$ to the suppliers\;
 \uIf{$\sum_{i = 1}^n x_{it}^*(p_t) \ge d$}{
 \tcc{\scriptsize If production exceeds demand, narrow search interval to left half of interval}
 Set $\S_p^{t+1} \gets [a_t, p_t]$, $a_{t+1} \gets a_t$, and $b_{t+1} \gets p_t$\;}
 \uElse{
 Set $\S_p^{t+1} \gets [p_t, b_{t}]$, $a_{t+1} \gets p_t$, and $b_{t+1} \gets b_t$\;
 }
 Set $p_{t+1} \gets \frac{a_{t+1}+b_{t+1}}{2}$ \;
 }
\caption{Feasible Price Set Tracking under Fixed Demand and Costs} 
\label{alg:fixed-demand-constant-reg}
\end{algorithm}

We now establish the main result of this section, which shows that Algorithm~\ref{alg:fixed-demand-constant-reg} achieves $O(1)$ regret on all three performance metrics.

\begin{theorem} \label{thm:IdenticalResult}
The unmet demand, cost regret, and payment regret of Algorithm~\ref{alg:fixed-demand-constant-reg} are $O(1)$ if the cost functions of the suppliers are strongly convex.
\end{theorem}

The proof of Theorem~\ref{thm:IdenticalResult} relies on the following key intermediate result, which is of independent interest, that establishes a unified regret metric that upper bounds the three performance metrics in Section~\ref{sec:perf-measures}.

\begin{proposition}[Unified Regret Upper Bound] \label{prop:unified-regret-bound}
    Define a unified metric that corresponds to the cumulative differences between the operator's posted prices and the equilibrium price: \vspace{-5pt}
    \begin{equation}
    \label{equ:unified-metric-price}
        R_T(\ppi) := \sum_{t=1}^T |p_t - p^*|.
    \end{equation}
    Then, there exists a constant $C>0$ such that this regret metric upper bounds (up to constant factors) the unmet demand, payment regret and cost regret metrics, i.e., \vspace{-5pt}
    \[\max\left\{U_T(\ppi), P_T(\ppi), C_T(\ppi)\right\} \le C \cdot R_T(\ppi).\]
\end{proposition}

Proposition~\ref{prop:unified-regret-bound} highlights that to establish Theorem~\ref{thm:IdenticalResult} and obtain bounds on each of our studied regret metrics in Section~\ref{sec:perf-measures}, it suffices to upper bound a single unified regret metric $R_T(\ppi)$. Note that such a unified regret metric in Equation~\eqref{equ:unified-metric-price}, which simultaneously upper bounds all three performance metrics, helps streamline our analysis of Theorem~\ref{thm:IdenticalResult}. For more details on the proof of Theorem~\ref{thm:IdenticalResult} and Proposition~\ref{prop:unified-regret-bound}, see Section~\ref{sec:proof-prop1}. However, we note that the metric in Equation~\eqref{equ:unified-metric-price} cannot be directly applied to the more general time-varying demand and cost settings studied in Sections~\ref{sec:vary-demand} and~\ref{sec:time-vary-cost-main}, respectively, which requires us to address some additional technical challenges, which we highlight in detail in the respective sections. 

We conclude this section with some remarks comparing our framework and results obtained in this section to the literature (Remark~\ref{rmk:comp-KL}) and more abstract root-finding problems (Remark~\ref{rmk:general-framework}). 

\begin{remark}[Comparison to~\citet{oppa}] \label{rmk:comp-KL}
Algorithm~\ref{alg:fixed-demand-constant-reg} is akin to the prior work of~\citet{oppa} as we also learn a feasible price interval at each step based on the gap between the observed production and customer demand. However, our setting differs from~\citet{oppa} in a critical way resulting in different guarantees. Specifically, in \cite{oppa}, achieving a regret below $O(\log \log T)$ is impossible, which stems from the discontinuity of the revenue function: it is zero until the price exceeds a buyer’s value, then jumps to that value. In contrast, we leverage the structure of our problem setting to establish Proposition~\ref{prop:unified-regret-bound}, which relies on proving important Lipschitzness relations for our studied regret metrics (see Section~\ref{sec:proof-prop1}), which enable us to achieve constant regret using a binary search procedure, unlike in the setting of \cite{oppa}.
\hfill //
\end{remark}

\begin{remark}[Extending Framework to Abstract Root-Finding Problem] \label{rmk:general-framework}
While we study the problem of equilibrium pricing inspired by electricity markets, we note that the binary search procedure used in Algorithm~\ref{alg:fixed-demand-constant-reg} is more broadly applicable to abstract root-finding problems of the form $g(p) = d$, where the unknown $g$ is monotonically increasing. In this light, our framework, and in particular Theorem~\ref{thm:IdenticalResult}, naturally extends to a more abstract class of problems involving learning the root of a monotonic equation.

Specifically, consider a setting in which $p^*$ is the unique solution to the equation $g(p) = d$, and an algorithm $\ppi$ sets a sequence of prices $p_1, \ldots, p_T$. Under the regret metric $\sum_{t = 1}^T |g(p_t) - d|$, which captures the cumulative deviation from the target value $d$, we observe that if 
$g$ is Lipschitz, then this regret can be upper bounded by the unified regret metric in Equation~\eqref{equ:unified-metric-price}, up to constant factors. As a result, Theorem~\ref{thm:IdenticalResult} directly implies a sublinear regret bound in this more abstract setting.

We note that our equilibrium pricing model can thus be viewed as an instantiation of this abstract framework, where $g(p) = \sum_{i = 1}^n x_i^*(p)$ represents total supplier production at price $p$, and each $x_i^*(p)$ is the solution to Problem~\eqref{eq:supObj}. In this regard, a key component of our analysis of Algorithm~\ref{alg:fixed-demand-constant-reg} is to establish Proposition~\ref{prop:unified-regret-bound} for our studied equilibrium pricing setting.

Finally, we note that while the abstract root-finding framework offers a natural generalization of the fixed cost and demand setting, our focus on equilibrium pricing in this work is more natural in the time-varying cost and demand settings, where the interpretation of the abstract root-finding framework when $g$ and $d$ vary over time becomes less intuitive. Note that time-varying supply costs and demand are fundamental features of dynamic equilibrium pricing models and, in particular, real-world electricity markets, and our framework and results in Sections~\ref{sec:vary-demand} and~\ref{sec:time-vary-cost-main} aim to capture these dynamics meaningfully.
\hfill //

\end{remark}

\

\subsection{Proof Details of Theorem~\ref{thm:IdenticalResult} and Proposition~\ref{prop:unified-regret-bound}} \label{sec:proof-prop1}
\
We establish the regret guarantee in Theorem~\ref{thm:IdenticalResult} by first showing the unified regret upper bound in Proposition~\ref{prop:unified-regret-bound}. The proof of Proposition~\ref{prop:unified-regret-bound} relies on two intermediary results. First, we establish a Lipschitzness relation between the optimal solution to Problem~\eqref{eq:supObj} and the prices set by the market operator, as is elucidated by the following lemma.
\

\begin{lemma} [Lipschitzness of Production in Prices] \label{lem:ProductionLipschitz}
Suppose that the suppliers' cost functions $c_i(\cdot)$ are $\mu$-strongly convex. Then, at any period $t$, the optimal production quantity for supplier $i$ corresponding to the solution of Problem~\eqref{eq:supObj} is Lipschitz in the price $p$, i.e., for all $p_1, p_2 \in [0, 1]$, $|x_{it}^*(p_1) - x_{it}^*(p_2)| \leq L_0 |p_1 - p_2|$ for some constant $L_0>0$ which only depends on $\mu$.
\end{lemma}

Next, we use Lemma~\ref{lem:ProductionLipschitz} to show that both the payment and cost regret metrics are Lipschitz in the prices. 
\
That is, Lemma~\ref{lem:lipschitzPricesRegret} establishes that small changes in the prices result in only small changes in the cost and payment regret metrics.
The proof of both lemmas can be found in Appendix~\ref{sec:lipschitz-lemma}.

\begin{lemma} [Lipschitzness of Regret Metrics in Prices] \label{lem:lipschitzPricesRegret}
Suppose that the suppliers' cost functions $c_i(\cdot)$ are $\mu$-strongly convex.
The payment and cost regret of an online pricing policy $\ppi$ are upper bounded by the absolute difference in prices $p_t$ corresponding to $\ppi$ and equilibrium prices $p^*_t$.
Namely, there exists some constant $L_1, L_2 > 0$ which only depend on $\mu$ and $n$ so that 
\[P_T(\ppi) \le L_1 \sum_{t=1}^T |p_t - p^*_t|, \text{and }  C_T(\ppi) \le L_2 \sum_{t=1}^T |p_t - p^*_t|. \]
\end{lemma}

By employing Lemmas~\ref{lem:ProductionLipschitz} and \ref{lem:lipschitzPricesRegret}, we conclude that we can simultaneously upper bound the three performance metrics by the unified regret metric in Equation~\eqref{equ:unified-metric-price}, as is highlighted by the following corollary.

\begin{corollary}
    \label{thm:unified-metric-price}
    The regret metric \eqref{equ:unified-metric-price} is an upper bound (up to constant factors) for the unmet demand, payment regret and cost regret metrics, so that 
    \[\max\left\{U_T(\ppi), P_T(\ppi), C_T(\ppi)\right\} \le \max\{L_0, L_1, L_2\} \cdot R_T(\ppi).\]
\end{corollary}

Note that Corollary~\ref{thm:unified-metric-price} establishes Proposition~\ref{prop:unified-regret-bound}, where the constant $C = \max\{L_0, L_1, L_2\}$. Then, to establish Theorem~\ref{thm:IdenticalResult}, it suffices to show that an algorithm $\ppi$ achieves sublinear regret in $R_T(\ppi)$.
With this unified regret metric in mind, we now present the proof of Theorem~\ref{thm:IdenticalResult}.

\begin{proof}[Proof of Theorem~\ref{thm:IdenticalResult}]
Recall from Corollary~\ref{thm:unified-metric-price} that to bound the unmet demand, cost regret, and payment regret metrics, it suffices to bound the quantity: $\sum_{t = 1}^T |p_t - p^*|$, where $p^*$ is the market-clearing price.

Denote $|S_p^t| = b_t - a_t$ as the length of the feasible price interval at each period $t$. Then, by the monotonicity of the optimal production quantity in the prices, we have that:
\begin{align*}
    \sum_{t = 1}^T |p_t - p^*| \leq \sum_{t = 1}^T |S_p^t| = \sum_{t = 1}^T \left(\frac{1}{2} \right)^{t} \stackrel{(a)}{=} \frac{0.5(1 - 0.5^T)}{0.5} \leq 1,
\end{align*}
where (a) follows from the formula for the sum of a geometric series. This establishes our desired result.
\end{proof}
\

\subsection{Remarks on Strong Convexity of Cost Functions}
\label{sec:fixed-convex-limitations}

\
While Algorithm 1 achieves sub-linear regret on all three performance metrics in the incomplete information setting for strongly convex cost functions, the strong convexity condition is crucial, and this result does not generalize to the setting of general convex costs. 
In particular, in this section, we show that if the cost functions of the suppliers are allowed to be arbitrarily close to a linear function, then no online algorithm can achieve sub-linear regret guarantees for the unmet demand, cost regret, and payment regret metrics simultaneously.
\
This result highlights the difficulty of the incomplete information setting compared to that with complete information, where the equilibrium price satisfying the three desirable properties in Section \ref{sec:market-model} exists and can be computed through the dual variable of the market clearing constraint of Problem \eqref{eq:supObj2}-\eqref{eq:demand-con} when the cost functions of all suppliers are convex (and, not necessarily, strongly convex).

\
\begin{proposition} \label{prop:impossibility-strongconvex}
    There exists a class of convex supplier cost functions for which no online algorithm can achieve sub-linear regret on all three performance metrics.
\end{proposition}

\begin{proof}
    We consider a single supplier with a fixed cost function of the form $c(x) = \frac{\alpha}{2} x^2 + \beta x$, where $\alpha, \beta > 0$.
    Then, this supplier's production when given a price $p$ satisfies \[p = c'(x^*(p)) = \alpha x^*(p) + \beta.\]
    So, the supplier's production is given by $x^*(p) = \frac{p-\beta}{\alpha}$ and the equilibrium price $p^* = \alpha d + \beta$ at demand $d$.
    
    For any given algorithm, let $p_0$ be its initial price offering before observing any responses from the supplier.
    Then, if $p_0 > p^*$, the payment regret can be written as
    \[p_0 \left(d + \frac{p_0-p^*}{\alpha}\right) - p^* d > p^* \cdot \frac{p_0 - p^*}{\alpha} = d(p_0 - p^*) + \frac{\beta}{\alpha} (p_0 - p^*).\]
    If the algorithm's choice on $p_0$ is positive, then because the cost function is unknown, for any fixed $\alpha$, we can choose $\beta = p_0/4$ and $d = \min\{\bar{d}, p_0/(4\alpha)\}$ so that $p^* < p_0/2$.
    Then, the payment regret is on the order of $O(1/\alpha)$, which is unbounded in the limit as $\alpha\to 0$.

    It follows that $p_0$ must be 0, but then the total production is 0 for any convex supplier cost function.
    This means that the operator learns nothing about the unknown cost function and we can repeat the above argument to show that the posted price at any period must be 0 to avoid arbitrarily large payment regret.
    Consequently, our analysis implies that the unmet demand is linear in time if the price is zero at each period; hence, it is not possible to achieve sublinear regret on all three regret metrics in the limit as $\alpha \to 0$.
\end{proof}

Given Proposition~\ref{prop:impossibility-strongconvex}, we assume that the cost functions are $\mu$-strongly convex for a universal constant $\mu > 0$ in the remainder of this work, which closely aligns with cost functions often studied in the context of electricity markets~\citep{TSAOUSOGLOU2022111890,1102918,7796713,4112652}.
\

\section{Fixed Cost Functions and Time-Varying Demand} \label{sec:vary-demand}

Having studied the setting of fixed supplier cost functions and customer demands over time, in this section, we investigate a more general market setting when the suppliers' cost functions are static while customer demands can vary across the $T$ periods. In particular, we suppose that the customer demands for the commodity are time-varying and lie in a continuous but bounded interval, i.e., the customer demand at each period $t$ is some variable quantity $d_t \in [\underline{d}, \overline{d}]$. In this setting, we extend the algorithm developed for fixed supplier cost functions and customer demands (Algorithm~\ref{alg:fixed-demand-constant-reg}) and show that it achieves a regret of \ $O(\sqrt{T})$ \ on all three performance metrics for strongly convex cost functions.

Our approach for the time-varying demand setting builds upon the algorithmic ideas for the fixed demand setting in Section~\ref{sec:fixed-setting}. To address the challenge that the demands can vary between the interval $[\underline{d}, \Bar{d}]$, we first consider a direct extension of Algorithm~\ref{alg:fixed-demand-constant-reg} to the time-varying demand setting, wherein a feasible price set is maintained for each realized demand. However, as there may be up to $O(T)$ different demand realizations over the $T$ periods, the worst-case regret of such an algorithm is $O(T)$. To resolve this issue, we leverage the intuition that customer demands that are close to each other correspond to equilibrium prices that are also close together. Thus, we uniformly partition the demand interval $[\underline{d}, \Bar{d}]$ into sub-intervals of width $\gamma$ and consider any demand in the same sub-interval the same. In particular, any demand lying in a given sub-interval, i.e., $d_t \in [\underline{d} + k \gamma, \underline{d} + (k+1) \gamma]$ for some $k \in \mathbb{N}$, is considered as a demand equal to the lower bound of that interval. Note then that from the perspective of the algorithm, there are $O(\frac{1}{\gamma})$ distinct demands, as opposed to $O(T)$ possible demand realizations, as the feasible demand interval is partitioned into $O(\frac{1}{\gamma})$ sub-intervals. Finally, for these $O(\frac{1}{\gamma})$ distinct demands, corresponding to the lower ends of the $O(\frac{1}{\gamma})$ sub-intervals, we apply the aforementioned direct extension of Algorithm~\ref{alg:fixed-demand-constant-reg}. 
These ideas are formally implemented in Algorithm~\ref{alg:time-varying-demand-helper-sqrt}, which will be used as a subroutine for the eventual algorithm (Algorithm~\ref{alg:time-varying-demand-new-sqrt}) we develop for this time-varying demand setting.

\begin{algorithm}
\SetAlgoLined
\SetKwInOut{Input}{Input}\SetKwInOut{Output}{Output}
\Input{Demand interval width $\gamma$, Start and end time $t_0, t_1$
}
Discretize demand intervals into $I_1, \dots, I_K$ with $I_k = [\underline{d} + (k-1)\gamma, \underline{d} + k\gamma)$ such that $K\gamma = \overline{d} - \underline{d}$\;
 Initialize a feasible price set $\S_{k} = [a_k, b_k] = [0, 1]$ and current price $p_k = \frac{a_k+b_k}{2} = \frac{1}{2}$ for each demand interval $I_{k}$\;
 \For{$t = t_0, t_0+1, \dots, t_1-1$}{
 Determine $k_t$ such that $d_t \in I_{k_t} =: [r_{k_t}, r_{k_t+1})$\;
 Offer price $p_{k_t}$ to the suppliers\;
 \uIf{$\sum_{i = 1}^n x_{it}^*(p_{k_t}) \ge r_{k_t}$}{
 Set $\S_{k_t} \gets [a_{k_t}, p_{k_t}]$, $b_{k_t} \gets p_{k_t}$ \;
 }
 \uElse{
 Set $\S_{k_t} \gets [p_{k_t}, b_{k_t}]$, $a_{k_t} \gets p_{k_t}$ \;
 }
Set $p_{k_t} \gets \frac{1}{2}(a_{k_t} + b_{k_t})$
 }
\caption{Feasible Price Set Tracking for Time-Varying Demands (helper function)}
\label{alg:time-varying-demand-helper-sqrt}
\end{algorithm}

\
For each of the three performance metrics, the regret incurred by Algorithm~\ref{alg:time-varying-demand-helper-sqrt} can be broken down into two parts: (i) the regret incurred by the Algorithm~\ref{alg:fixed-demand-constant-reg} sub-routine on each of the demand sub-intervals, and (ii) the inaccuracies of considering all demands in a given sub-interval to be equal to the lower end of that sub-interval.
For our choice of the demand sub-interval width $\gamma$, a larger value of $\gamma$ leads to smaller regret for the first part, and a smaller value of $\gamma$ implies that we incur a smaller per-step discretization error for the second part.
Thus, in Algorithm~\ref{alg:time-varying-demand-helper-sqrt}, we want to choose an appropriate value of $\gamma$ that balances these two sources of regret. Then, to apply the procedure in Algorithm~\ref{alg:time-varying-demand-helper-sqrt} without pre-determined knowledge on the time horizon $T$, Algorithm~\ref{alg:time-varying-demand-new-sqrt} repeatedly invokes Algorithm~\ref{alg:time-varying-demand-helper-sqrt} with the value of $\gamma$ shrinking over time.
\

\begin{algorithm}
\SetAlgoLined
\SetKwInOut{Input}{Input}\SetKwInOut{Output}{Output}
 \For{$m = 0, 1, \dots$}{
    Execute Algorithm~\ref{alg:time-varying-demand-helper-sqrt} with $t_0 = 2^m, t_1 = 2^{m+1}$ and $\gamma = 1/\sqrt{t_1 - t_0} = 2^{-m/2}$.
 }
\caption{Feasible Price Set Tracking for Time-Varying Demands (main loop)}
\label{alg:time-varying-demand-new-sqrt}
\end{algorithm}
We now present the main result of this section, which establishes that Algorithm~\ref{alg:time-varying-demand-new-sqrt} achieves a regret of \ $O(\sqrt{T})$ \ for strongly convex cost functions of suppliers. 

\begin{theorem} \label{thm:VaryingDemandResult-sqrtT}
The unmet demand, cost regret, and payment regret of Algorithm~\ref{alg:time-varying-demand-new-sqrt} are $O(\sqrt{T})$ if the cost functions of the suppliers are strongly convex.
\end{theorem}

\begin{proof}[Proof (sketch)]
To better illuminate the key ideas behind the proof, we show a simpler claim that Algorithm~\ref{alg:time-varying-demand-helper-sqrt} incurs \ $O(\sqrt{T})$ \ regret when we choose $t_0 = 1, t_1 = T$, and $\gamma = 1/\sqrt{T}$.

We recall that the regret incurred by Algorithm~\ref{alg:time-varying-demand-helper-sqrt} has two parts: (i) the regret incurred by the Algorithm~\ref{alg:fixed-demand-constant-reg} sub-routine for each demand sub-interval, and (ii) the inaccuracies of considering all demands in a given sub-interval to be equal to the lower end of that sub-interval.
First, by invoking the same binary search argument as Theorem~\ref{thm:IdenticalResult}, the first part is of order $O(K)$ for all three regret metrics, where $K := \lceil(\overline{d} - \underline{d})/\gamma\rceil$.

For the second part, since all demands in a given sub-interval are treated as a demand equal to the lower bound of that sub-interval and the suppliers' optimal production is monotonic in the price, every price $p_t$ offered by Algorithm~\ref{alg:time-varying-demand-helper-sqrt} is an under-estimate to the equilibrium price for demand $d_t$.
Thus, the second part of the regret is only positive for the unmet demand and is at most $O(\gamma T)$, as the width of each demand sub-interval is $\gamma$, and regret is only accumulated over $T$ periods. 
Finally, choosing $\gamma = \frac{1}{\sqrt{T}}$ achieves an optimal balance between the two quantities above, i.e., $O(\gamma^{-1})$ and $O(\gamma T)$, establishing the $O(\sqrt{T})$ regret bound.

As for the regret of Algorithm~\ref{alg:time-varying-demand-new-sqrt}, we can apply the argument above to each invocation of Algorithm~\ref{alg:time-varying-demand-helper-sqrt} and sum up the regrets incurred by each of these episodes.
This results in an \ $O(\sqrt{T})$ \ regret bound in all three performance metrics.
\end{proof}

\
We note that, in our proof of Theorem~\ref{thm:VaryingDemandResult-sqrtT}, we directly analyzed each of the performance metrics instead of appealing to the unified metric \eqref{equ:unified-metric-price}.
If we directly apply the unified metric~\eqref{equ:unified-metric-price}, we seek to show that
\begin{equation}
\label{equ:varying-demand-unified}
    \sum_{t=t_0}^{t_1-1} |p_t - p^*(d_t)| \le O(\sqrt{T}),
\end{equation}
where $p^*(d)$ is the equilibrium price for demand $d$.
However, for each demand $d_t$, Algorithm~\ref{alg:time-varying-demand-helper-sqrt} searches for the price that matches a different demand $r_{k_t} \le d_t$.
Therefore, the binary search subroutine which we employ in Algorithm~\ref{alg:time-varying-demand-helper-sqrt} results in a bound that is different from \eqref{equ:varying-demand-unified}:
\begin{equation}
\label{equ:varying-demand-part1}
    \sum_{t=t_0}^{t_1-1} |p_t - p^*(r_{k_t})| \le O(K),
\end{equation}
where we discretize the demands into $K$ intervals.
To connect the two expressions \eqref{equ:varying-demand-unified} and \eqref{equ:varying-demand-part1}, we would need to determine the sensitivity of the equilibrium price $p^*(\cdot)$ as a function of the demands to bound the quantity $|p^*(d_t) - p^*(r_{k_t})|$.
But deriving this sensitivity bound would lead to suboptimal constant factors in our regret guarantees and requires an additional Lipschitz-smoothness assumption on the cost functions that is not required in the proof of Theorem~\ref{thm:VaryingDemandResult-sqrtT}. This additional assumption is used in Lemma~\ref{lem:PriceLipschitz}, which is exclusively utilized in the time-varying cost setting in Section~\ref{sec:time-vary-cost-main}. 
In the time-varying demand setting studied in this section, we do not introduce an additional Lipschitz smoothness assumption to present Theorem~\ref{thm:VaryingDemandResult-sqrtT} in its most general form.
Instead, we obtain tighter guarantees for Algorithm~\ref{alg:time-varying-demand-helper-sqrt} by directly relating each of the three performance metrics to the expression in \eqref{equ:varying-demand-part1}.
\

For a complete proof of Theorem~\ref{thm:VaryingDemandResult-sqrtT}, see Appendix~\ref{sec:vary-demand-pf}. We reiterate that Theorem~\ref{thm:VaryingDemandResult-sqrtT} applies to strongly convex cost functions as with Theorem~\ref{thm:IdenticalResult} and that extending this result to general convex cost functions, e.g., linear functions, is, in general, not possible (see Section~\ref{sec:fixed-convex-limitations}). Furthermore, compared to the regret guarantee obtained in Proposition~\ref{thm:IdenticalResult}, Theorem~\ref{thm:VaryingDemandResult-sqrtT} establishes that the time-varying nature of the customer demand incurs an additional factor of $O(\sqrt{T})$ in the regret guarantee as compared to the setting with fixed demands. However, we do note that if the set of demand realizations $D$ is known \emph{a priori} to be $o(\sqrt{T})$, then the regret guarantee in Theorem~\ref{alg:time-varying-demand-new-sqrt} can be improved to $O(|D|)$ by running the direct extension of Algorithm~\ref{alg:fixed-demand-constant-reg}, wherein a feasible price interval is maintained for each realized demand. Finally, we note that the regret guarantee obtained in Theorem~\ref{thm:VaryingDemandResult-sqrtT} compares favorably to classical $O(\sqrt{T})$ regret guarantees in the OCO or MAB literature~\citep{OPT-013}.

\section{Time-Varying Cost Functions} \label{sec:time-vary-cost-main}
In this section, we consider the general setting, where, in addition to customer demands changing over time, suppliers' cost functions can also vary across the $T$ periods. Formally, at each period $t$, each supplier $i$ has a privately known and time-varying cost function $c_{it}(\cdot)$. Compared to the setting with fixed cost functions, the fundamental ideas underlying the performance guarantees of Algorithms~\ref{alg:fixed-demand-constant-reg} and \ref{alg:time-varying-demand-new-sqrt} do not directly apply to this setting, as suppliers' production may not remain the same when the operator offers the same price at different periods due to the time-varying nature of their cost functions. In fact, in Section~\ref{sec:counter-example}, we show that if the operator does not know the process that governs the variation of the cost functions, then sub-linear regret is impossible to achieve on all three regret metrics.
To this end, in Section~\ref{sec:contextual-bandit}, in alignment with real-world markets, we consider an augmented problem setting wherein the market operator is provided with a hint (i.e., context) on the variation in suppliers' cost functions over time, e.g., due to weather conditions in electricity markets, while still keeping the full description of the costs away from the operator. In this setting, where the operator has access to additional hints on suppliers' cost functions, we then develop an algorithm with sub-linear regret on all three performance metrics (see Sections~\ref{sec:igw-alg-sol} and~\ref{sec:igw-alg-sketch}) through an adaptation of an algorithm in the contextual bandits literature~\citep{foster2020beyond}.

\subsection{Impossibility of Setting Equilibrium Price Under Time-Varying Costs}
\label{sec:counter-example}

We initiate our study of the setting of time-varying costs by presenting an example that illustrates the impossibility of setting equilibrium prices if the market operator has no information on how the cost functions of suppliers change over time. In particular, Proposition~\ref{prop:time-varying-cost-countereg} presents a counterexample establishing that even if suppliers' cost functions are drawn i.i.d. from a known distribution, no online algorithm can achieve sub-linear regret on all three regret metrics as long as the operator is not informed about the outcome of the random draws from the distribution.

\begin{proposition} [Impossibility of Sub-linear Regret for Time-varying Costs] \label{prop:time-varying-cost-countereg}
There exists an instance with fixed time-invariant demand and a single supplier whose cost functions are drawn i.i.d. from some (potentially known) distribution such that no online algorithm can achieve sub-linear regret on all three regret metrics.
\end{proposition}

\begin{proof}[Proof (sketch):]
We consider a setting with a fixed demand of $d=1$ at every period and a single supplier whose cost functions at each period are drawn from a distribution such that at each period $t$, its cost function could be either $c_1(x) = \frac{1}{8}x^2$ or $c_2(x) =\frac{1}{16} x^2$, each with probability $0.5$. We suppose that the market operator has knowledge of the distribution from which the supplier's cost function is sampled i.i.d. but does not know the outcome of the random draw at any period and show that any pricing strategy adopted by the operator must incur a linear regret on at least one of the three regret metrics for this instance.
To this end, we analyze the total regret, i.e., the sum of the unmet demand, payment regret, and cost regret and show that irrespective of the set price $p$ at any period $t$, the expected total regret at any period is at least $\frac{7}{64}$, i.e., the total regret is at least $\frac{7}{64} T$. Since the sum of the three regret metrics is linear in $T$, at least one of the three metrics must be linear in $T$, establishing our claim.
\end{proof}

For a complete proof of Proposition~\ref{prop:time-varying-cost-countereg}, see Appendix~\ref{sec:pf-prop-countereg}. While it was possible to achieve sub-linear regret in the setting with time-varying customer demands (see Theorem~\ref{thm:VaryingDemandResult-sqrtT}), Proposition~\ref{prop:time-varying-cost-countereg} establishes that such a result is, in general, not possible in the setting with time-varying cost functions. The setting with time-varying cost functions is more challenging because the market operator observes customer demands, which it can use to make pricing decisions, but does not observe the cost functions of suppliers.
Thus, Proposition~\ref{prop:time-varying-cost-countereg} shows that it is impossible to jointly optimize the three performance metrics, and illustrates the difficulty of balancing the three performance metrics, as decreasing the payment or cost regret causes an increase in the unmet demand and vice versa.

\subsection{Adding Contexts for Time-varying Costs}
\label{sec:contextual-bandit}

Proposition~\ref{prop:time-varying-cost-countereg} highlights that if the operator does not have any information on the change in suppliers' cost functions over time, it is impossible to achieve sub-linear regret on all three regret metrics. To this end, in this section, we consider a natural augmented problem setting wherein the market operator, without knowing the complete specification of cost functions of suppliers, additionally has access to a hint (i.e., context) that reflects the variation in cost functions of suppliers over time. We note that such a setting aligns with real-world markets, e.g., electricity markets, wherein the cost functions of suppliers are private information yet typically vary over time based on observed quantities, such as changes in the ambient weather conditions.  

To specify the augmented problem setting with contexts, we first introduce some notation regarding the cost functions of suppliers. In particular, we assume that each supplier's cost function is composed of two parts: (i) an unknown component that is time-invariant, and (ii) a time-varying component that is revealed to the market operator. More precisely, the cost function of each supplier $i$ is parameterized as follows:
\[c_{it}(\cdot) = c_i(\cdot; \phi_i, \theta_{it}),\]
where $\phi_i$ is private information and $\theta_{it}$ is the time-varying component of the cost function given to the operator as \textit{contexts}.
Note that for any fixed $\phi_i$, the context $\theta_{it}$ uniquely determines the cost function of supplier $i$ at time $t$.
We stress that we do not assume any structure on the parameterization of the cost functions, and so the time-varying and time-invariant components of the cost functions need not be separable.
Further, since $\phi_i$'s are unknown, the market operator cannot directly solve Problem~\eqref{eq:supObj2}-\eqref{eq:demand-con} to obtain the equilibrium prices in the market.

For the simplicity of exposition, for the remainder of this section, we aggregate all suppliers' cost functions into a combined cost $c_t(\cdot; \theta_t) = \sum_{i=1}^n c_{it}(\cdot; \phi_i, \theta_{it})$, where $\theta_t = (\theta_{1t}, \dots, \theta_{nt})$ is the time-varying context associated with the combined cost function. 
Note that doing so is without loss of generality, as all suppliers have convex costs and observe the same prices in the market.
Furthermore, we note that since the private information $\phi_1, \dots, \phi_n$ are unknown, the market operator cannot directly solve Problem~\eqref{eq:supObj2}-\eqref{eq:demand-con} to obtain the equilibrium prices in the market.

In this augmented problem setting, at each period $t$, in addition to receiving the customer demand $d_t$, the market operator observes a context $\theta_t$, which it can use along with the prior history of supplier production quantities, customer demands, and contexts, to set a price $p_t$. In particular, with access to sequentially arriving contexts, the market operator sets a sequence of prices given by the pricing policy $\ppi = (\pi_1, \ldots, \pi_T)$, where $p_t = \pi_t(\{ (x_{t'}^*)_{1 = 1}^n, d_{t'}, \theta_{t'} \}_{t'=1}^{t-1}, d_t, \theta_t)$, where $x_{t}^*$ represents the sum of optimal production quantity corresponding to the solution of Problem~\eqref{eq:supObj} for each supplier at period $t$. We then evaluate the performance of this class of pricing policies on three regret metrics introduced in Section~\ref{sec:perf-measures}.
Note that we can naturally extend these three metrics to the augmented setting with contexts by plugging in $c_{it}(\cdot) = c_i(\cdot; \phi_i, \theta_{it})$ and for completeness, we present the corresponding definitions explicitly in Appendix~\ref{sec:new-regret-defs}.

\subsection{Algorithm for Time-Varying Costs with Contexts}
\label{sec:igw-alg-sol}

We now present an algorithm that simultaneously achieves sub-linear regret for the unmet demand, payment regret, and cost regret metrics for the augmented problem setting introduced in Section~\ref{sec:contextual-bandit}. Our algorithmic approach is inspired by recent ideas in the contextual bandits literature (e.g.~\cite{agarwal2014taming,foster2020beyond}) and involves two building blocks. 
First, we seek to learn how to associate the arriving contexts with the relevant properties of suppliers' cost functions. Next, based on the information inferred from the arriving contexts, our algorithm offers prices to suppliers in the next period.
We note that the first step of our approach fundamentally differs from the fixed cost and time-varying demand setting in Section~\ref{sec:vary-demand}, as time-varying supplier costs, unlike time-varying demands, are not observed and thus unavailable to the operator making pricing decisions.
Therefore, the first step of our approach is crucial in learning a descriptive model on the supplier's response to various contexts and prices.

The task of learning to associate the arriving contexts with the relevant properties of the cost function is accomplished by an \textit{online regression oracle}. In particular, an online regression oracle performs real-valued online regression and achieves a prediction error guarantee, with a bound denoted $\est(T)$, relative to the best function in a class $\mathcal{F}$.

\begin{definition}[Online Regression Oracle]
Consider a function class $\mathcal{F} : \mathcal{A} \to \mathcal{B}$, at each time $t$, the online regression oracle receives an input $a_t$ and computes an estimate $\hat{b}_t = \hat{f}_t(a_t)$, where $\hat{f}$ depends on the past history $\mathcal{H}_{t-1} = (a_1, b_1), \dots, (a_{t-1}, b_{t-1}).$
Then, the oracle receives the true output $b_t$.

The predictors $\hat{f}_t$ of the oracle are almost as accurate as any function in $\mathcal{F}$ in the sense that:
\[\sum_{t=1}^T (\hat{b}_t - b_t)^2 - \inf_{f \in \mathcal{F}} \sum_{t=1}^T (f(a_t) - b_t)^2 \le \est(T).\]
\end{definition}

The prediction error $\est(T)$ of the online regression oracle scales with the ``size'' of the $\mathcal{F}$ in a statistical sense. As an example, if the function class $\mathcal{F}$ is finite, then the exponential weights update algorithm achieves $\est(T) \le \log |\mathcal{F}|$ (see e.g.~\cite{vovk1995game}).
Therefore, the function class $\mathcal{F}$ should be rich enough to capture the map between contexts and the variation in cost functions and also not be too large so that the estimation error could be small.

To construct a function class $\mathcal{F}$ appropriate for our problem setting, we first note that the market operator cannot observe the supplier's costs but can instead use the oracle to regress on the supplier's production $x^*(\cdot; \theta_t)$, which is directly observable.
Note that by the first-order optimality condition on Problem~\eqref{eq:supObj} that the production and price are related as follows:
\[c_t'(x^*(p; \theta_t); \theta_t) = p \implies x^*(p; \theta_t) = (c_t')^{-1}(p; \theta_t).\]
Note that when the cost functions are strongly convex (i.e., $c_t'$ is invertible), the production level $x^*(\cdot; \theta_t)$ is well-defined as a function of the price $p_t$ and context $\theta_t$.
Thus, we define the function class $\mathcal{F}$ as the possible mappings from the price-context tuple $(p, \theta)$ to the production $x^*$, i.e., the oracle tries to determine the amount of the supplier's production given a price $p_t$ and context $\theta_t$.

With our choice of the online regression oracle, we now present our algorithm for the setting of time-varying costs with contexts based on the inverse gap weighing method introduced in~\cite{foster2020beyond}. 
\
Because the set of permissible prices is infinite, we pick a finite set of $K$ prices that are uniformly spaced on the interval $[0, 1]$, so that our algorithm only chooses between these $K$ discrete prices.
The performance of our algorithm will depend on the choice of $K$, and we shall discuss the choice of an appropriate value of $K$ later in this section.
\
Given the oracle's output $\hat{f}_t$ that estimates the production quantity $x^*(p_t; \theta_t)$, the greedy choice at each period $t$ is to match the requested demand $d_t$ as closely as possible, i.e. choose a price $\hat{p}_t$ such that the quantity $|\hat{f}_t(\hat{p}_t; \theta_t) - d_t|$ is minimized.
To balance between both exploration and exploitation, Algorithm~\ref{alg:time-varying-cost} instead samples each price $p_t$ from the set of $K$ discrete prices according to a probability distribution $\Delta_t$. 
In effect, Algorithm~\ref{alg:time-varying-cost} achieves good exploration by choosing any of the $K$ prices with some positive probability, which ensures that the online regression oracle has access to a wide-ranging history $\mathcal{H}_{t-1}$, so it can achieve a low prediction error even when the market operator receives a new context or demand. Furthermore, to minimize the penalty for exploration, the probability distribution $\Delta_t$ is chosen such that it assigns the highest probability to the greedy choice $\Hat{p}_t$ under the current oracle estimate $\hat{f}_t$ while assigning a probability to every other price that is roughly inversely proportional to the gap between its unmet or excess demand and that of the greedy choice $\Hat{p}_t$. Then, for each period $t$, given this choice of $\Delta_t$, a price $p_t$ is sampled from this distribution, following which suppliers produce an optimal quantity $x^*(p_t; \theta_t)$ of the commodity as given by the solution of Problem~\eqref{eq:supObj}. Finally, the oracle is updated with the new context $\theta_t$, customer demand $d_t$, and optimal supplier production to generate a new estimator $\Hat{f}_{t+1}$ for the next period. This process is presented formally in Algorithm~\ref{alg:time-varying-cost}.

\begin{algorithm}
\SetAlgoLined
\SetKwInOut{Input}{Input}\SetKwInOut{Output}{Output}
\Input{Online regression oracle $\mathcal{O}$ with input pairs $(\theta_t, p_t)$ and output $x_t$ and exploration parameter $\gamma > 0$}
Define uniform $K$-cover of possible prices $0 = p_1 < p_2 < \cdots < p_K = 1$\;
 \For{$t = 1, \dots, T$}{
 Query the oracle $\mathcal{O}$ for an estimator $\hat{f}_t$\;
 Receive context $\theta_t$ and demand $d_t$\;
 Sample price $p_t$ from the probability distribution 
    \[\Delta_t(p_i) = \frac{1}{\lambda + 2 \gamma \left(|\hat{f}_t(p_i; \theta_t) - d_t| - |\hat{f}_t( \hat{p}_t; \theta_t) - d_t|\right)},\]
    with $\hat{p}_t = \argmin_{p \in \{p_1, \dots, p_K\}} |\hat{f}_t(p; \theta_t) - d_t|$ and $\lambda \in (0, K)$ as the normalization constant\;
Commit $p_t$ and observe the production $x_t = x^*(p_t; \theta_t)$ corresponding to the solution of Problem~\eqref{eq:supObj} given the price $p_t$\; 
Update the oracle $\mathcal{O}$ with $((\theta_t, p_t), x_t)$\;
}
\caption{Online Equilibrium Pricing for Time-Varying Costs}
\label{alg:time-varying-cost}
\end{algorithm}

We now present the main result of this section, which establishes that for an appropriate choice of the discretization $K$ of the price set, Algorithm~\ref{alg:time-varying-cost} can achieve sub-linear regret on all three regret metrics, as is elucidated through the following theorem. 
We highlight that this theorem holds for any sequence of contexts $\theta_t$, which means that $\theta_t$ can be derived from some physical dynamics, drawn from a probability distribution, or even chosen adversarially.

\begin{theorem}[Informal]
    \label{thm:igw-bound-informal}
    With high probability, for any sequence of contexts $\theta_t$ and customer demands $d_t$ over $T$ periods, the unmet demand, payment regret, and cost regret satisfy:
    \begin{equation} 
    \label{equ:informal-regret-bound}
        \EE_{p_t \sim \Delta_t, t \in [T]}[U_T(p_1, \dots, p_T)] \le O\left(\sqrt{KT \cdot \est(T)} + \frac{T}{K}\right),
    \end{equation}
    where $K$ is the number of prices in the uniformly discretized price set.
    \
    If we choose $K = \sqrt[3]{T / \est(T)}$, then we have
        \[ \EE_{p_t \sim \Delta_t, t \in [T]}[U_T(p_1, \dots, p_T)] \le O\left(T^{2/3} \cdot \sqrt[3]{\est(T)}\right).\]
    \
    And similar conclusions hold for $\EE_{p_t \sim \Delta_t, t \in [T]}[P_T(p_1, \dots, p_T)]$ and $\EE_{p_t \sim \Delta_t, t \in [T]}[C_T(p_1, \dots, p_T)]$.
\end{theorem}

We note that because the pricing policy in Algorithm~\ref{alg:time-varying-cost} is probabilistic, we present the regret bounds in expectation with respect to the distributions $\Delta_t$.
For the simplicity of exposition, we only present Theorem~\ref{thm:igw-bound-informal} as an informal statement and present the analysis of a rigorous theorem statement, which involves introducing additional notation, e.g., quantifying the high probability bound, in Appendix~\ref{sec:igw-alg-proof}. Furthermore, we present a sketch of the main ideas used to analyze Algorithm~\ref{alg:time-varying-cost} in Section~\ref{sec:igw-alg-sketch}.

\
The regret bound in Equation~\eqref{equ:informal-regret-bound} has two parts --- the first part corresponds to the regret bound of our algorithm against the offline optimal over a finite decision space, and the second part corresponds to the discretization error from only considering a finite set of prices from the continuous interval of permissible prices $p_t \in [0, 1]$.
The first part depends on the quantity $\est(T)$, which captures the ``size'' of the cost function class $\mathcal{F}$, so that a richer $\mathcal{F}$ is more difficult to estimate and results in weaker regret bounds.
For the second part of the regret bound, we note that the space of permissible actions (i.e., the prices $p_t$) is continuous, but most bandit algorithms can only keep track of finitely many distinct actions.
To address this difficulty, we restrict Algorithm~\ref{alg:time-varying-cost}'s actions to $K$ different prices.
This discretization of the prices leads to an error on the order of $1/K$.
Then, to get the final regret guarantee, we pick an appropriate value of $K$ to balance the regret incurred by the bandit algorithms on the finite set of prices and the discretization error.

This discretization technique is common in the literature for bandit problems with infinite action spaces.
For instance, in Lipschitz bandits (without contexts) where the action space is continuous and the reward function is Lipschitz in the actions, after restricting to a set of $K$ uniformly spaced actions, the regret incurred by a bandit algorithm, e.g., EXP3, is on the order of $\sqrt{TK}$ and the discretization error is on the order of $\frac{T}{K}$. After picking the optimal choice of $K = T^{1/3}$, the overall regret for a one-dimensional Lipschitz bandit is on the order of $T^{2/3}$, which is known to be optimal for the Lipschitz bandit setting~\citep{slivkins2011contextual,foster2021statistical}.

\begin{remark} \label{rmk:t^2/3}
When the function class $\mathcal{F}$ is finite, e.g., there are only finitely many possible cost functions, the multiplicative weights update algorithm can achieve $\est(T) = \log |\mathcal{F}|$.
Therefore, picking $K = \sqrt[3]{T / \log |\mathcal{F}|}$ achieves a sublinear regret bound of $O(T^{2/3} \sqrt[3]{\log |\mathcal{F}|})$.

Recall that, as previously discussed, the optimal regret for the Lipschitz bandit problem is on the order of $T^{2/3}$.
Therefore, the $O(T^{2/3})$ regret established in Theorem~\ref{thm:igw-bound-informal} for finite function classes (see Remark~\ref{rmk:t^2/3}) is consistent with known lower bounds in bandit settings with infinite action spaces, indicating that our result matches the expected regret scaling in such settings.
\hfill //
\end{remark}
\

Additionally, in Appendix~\ref{sec:igw-alg-proof}, we discuss several other function classes $\mathcal{F}$, including parametric classes, bounded Lipschitz functions, and neural networks with bounded spectral norm, for which Algorithm~\ref{alg:time-varying-cost} achieves a sub-linear regret guarantee by leveraging results from the statistical learning literature (see, e.g., \cite{wainwright2019high,rakhlin2014online}).

\subsection{Sketch of Main Ideas to Analyze Algorithm~\ref{alg:time-varying-cost}}
\label{sec:igw-alg-sketch}

\
In this section, we describe the key ideas in the analysis of Algorithm~\ref{alg:time-varying-cost}.
First, we introduce the following proxy regret metric and show that it is an upper bound for all three performance metrics.
\begin{equation}
\label{equ:unified-metric-prod}
    \textsc{Reg}(T) := \sum_{t=1}^T \EE_{p_t \sim \Delta_t} \left[\left|x^*(p_t; \theta_t) - d_t\right| \mid \mathcal{H}_{t-1} \right],
\end{equation}
where $x^*(p_t; \theta_t)$ is the optimal production quantity of the suppliers and $p_t$ is chosen according to the distribution $\Delta_t$ as defined in Algorithm~\ref{alg:time-varying-cost}. 
Note that the choice of distribution $\Delta_t$ depends on the past history $\mathcal{H}_{t-1} = ((p_1, \theta_1), x_1), \dots ((p_{t-1}, \theta_{t-1}), x_{t-1})$.
Then, the corresponding bandit reward function takes the form of $p  \mapsto |x^*(p; \theta_t) - d_t|$.

To show that the metric \eqref{equ:unified-metric-prod} upper bounds all three performance metrics, we establish the following lemma, which studies the solution of Problem~\eqref{eq:supObj2}-\eqref{eq:demand-con} and establishes the Lipschitzness of the optimal prices as a function of customer demands.
The proof of this lemma can be found in Appendix~\ref{sec:lipschitz-lemma}.

\begin{lemma} [Lipschitzness of Prices in Demands] \label{lem:PriceLipschitz}
Suppose that the suppliers’ cost functions $c_t(\cdot)$ are $\mu$-strongly convex and $\ell$-Lipschitz smooth.
The optimal prices corresponding to the dual variables of the market clearing constraint of Problem~\eqref{eq:supObj2}-\eqref{eq:demand-con} are Lipschitz in the demand $d$, i.e., for all $d_1, d_2 \in [\underline{d}, \Bar{d}]$, $|p^*(d_1) - p^*(d_2)| \leq L_3 |d_1 - d_2|$ for some constant $L_3>0$ which only depends on $\mu$ and $\ell$. Here, $p^*(d)$ is the optimal price corresponding to the dual variable of the market-clearing constraint of Problem~\eqref{eq:supObj2}-\eqref{eq:demand-con} with a customer demand of $d$.
\end{lemma}

We note that, unlike Lemmas~\ref{lem:ProductionLipschitz} and~\ref{lem:lipschitzPricesRegret}, Lemma~\ref{lem:PriceLipschitz} requires the additional assumption that supplier cost functions are Lipschitz smooth. This stronger assumption is crucial in establishing that the proxy regret metric defined in~\eqref{equ:unified-metric-prod} upper bounds the three regret metrics studied in this work (see discussion below). To present our results in their most general form in the fixed cost settings considered in Sections~\ref{sec:fixed-setting} and~\ref{sec:vary-demand}, where such Lipschitz smoothness assumptions are not imposed, we do not use the metric~\eqref{equ:unified-metric-prod} and instead present our results under more general cost function assumptions in those sections.

By employing Corollary~\ref{thm:unified-metric-price} and Lemma~\ref{lem:PriceLipschitz}, we can show that the proxy regret \eqref{equ:unified-metric-prod} upper bounds the unmet demand, payment regret, and cost regret metrics.
To see this, first note that this proxy regret upper bounds the unmet demand because
\[\left|x^*(p_t; \theta_t) - d_t\right| \ge \left(x^*(p_t; \theta_t) - d_t\right)_+.\]
Next, let $p^*_t$ be the equilibrium price for demand $d_t$, and note that $p_t$ is the equilibrium price when demand is equal to $x^*(p_t; \theta_t)$.
So, by Lemma~\ref{lem:PriceLipschitz}, we have that
\[ |p_t  - p^*_t| \le L_3 \left|x^*(p_t; \theta_t) - d_t\right|. \] 
And by applying Corollary~\ref{thm:unified-metric-price} to the previous inequality, we conclude that 
\[\EE_{p_t \sim \Delta_t, t \in [T]}[P_T] \le L_1 \sum_{t=1}^T \EE_{p_t \sim \Delta_t} \left|p_t - p^*_t\right| \le L_1 L_3 \sum_{t=1}^T \EE_{p_t \sim \Delta_t} \left|x^*(p_t; \theta_t) - d_t\right| = L_1 L_3 \textsc{Reg}(T).\]
Using a similar line of reasoning, we can also show that the proxy regret serves as an upper bound, up to constants, for the cost regret $\EE_{p_t \sim \Delta_t, t \in [T]}[C_T]$.

We note that the proxy regret metric \eqref{equ:unified-metric-prod} is different from the unified regret metric \eqref{equ:unified-metric-price} from Section~\ref{sec:fixed-setting} in that \eqref{equ:unified-metric-prod} measures the mismatch between the actual productions and demands, whereas \eqref{equ:unified-metric-price} measures the difference between our algorithm's prices and the equilibrium price.
In the current setting of time-varying cost functions, we note that the metric \eqref{equ:unified-metric-prod} is exactly equal to the regret incurred by a contextual bandit with a reward function $|f_t(p_t; \theta_t) - d_t|$.
In contrast, we cannot directly derive a regret bound in the form of the unified regret metric \eqref{equ:unified-metric-price} because, given that the equilibrium prices $p^*(d_t)$ are not observed by the algorithm, \eqref{equ:unified-metric-price} cannot be written as the sum of some bandit reward functions over time. Hence, we use the proxy regret metric \eqref{equ:unified-metric-prod} for the analysis of Algorithm~\ref{alg:time-varying-cost} in the time-varying cost setting.
\

Given the above observation that the proxy regret \eqref{equ:unified-metric-prod} is an upper bound on the three regret metrics, it suffices to show that Algorithm~\ref{alg:time-varying-cost} achieves a sub-linear regret, as is elucidated by the following proposition.
\begin{proposition}[informal]
\label{thm:igw-regret-informal}
    With high probability, for any sequence of context $\theta_t$ and demand $d_t$, the proxy regret is bounded by
    \[ \textsc{Reg}(T) \le O\left(\sqrt{KT \cdot \est(T)} + \frac{T}{K}\right).\]
\end{proposition}
\
Theorem~\ref{thm:igw-bound-informal} follows from Proposition~\ref{thm:igw-regret-informal} as the ``proxy'' regret serves as an upper bound on the unmet demand, payment regret, and cost regret due to the derived Lipschitz relations.
\
For a complete proof of Proposition~\ref{thm:igw-regret-informal}, see Appendix~\ref{sec:igw-alg-proof}.

\section{Numerical Experiments} \label{sec:experiments}

In this section, we present numerical experiments to corroborate several of the theoretical results obtained in the paper. In particular, we first validate our obtained regret bounds for Algorithms~\ref{alg:fixed-demand-constant-reg} and~\ref{alg:time-varying-demand-new-sqrt} in the setting with fixed costs and fixed demands and fixed cost but variable demands settings, respectively, in Section~\ref{sec:validate-algos1-2}. We then validate Proposition~\ref{prop:time-varying-cost-countereg}, which establishes the impossibility of developing an algorithm with sub-linear regret in the variable cost function setting, by numerically showing that a widely-studied primal-dual algorithm used in online learning literature, based on performing \emph{online dual sub-gradient descent}, incurs linear regret on at least one of the three performance metrics in our studied problem setting.

\subsection{Validating Theoretical Bounds of Algorithms~\ref{alg:fixed-demand-constant-reg} and~\ref{alg:time-varying-demand-new-sqrt}} \label{sec:validate-algos1-2}

In this section, we validate our regret bounds obtained in Theorems~\ref{thm:IdenticalResult} and~\ref{thm:VaryingDemandResult-sqrtT} for the fixed cost and fixed demands and fixed cost but variable demands settings, respectively.

\
\emph{Experiment Setup:} 
For the experiments presented in this section, we apply Algorithms~\ref{alg:fixed-demand-constant-reg} and~\ref{alg:time-varying-demand-new-sqrt} to two market instances, described below.

\emph{Instance 1:} We consider a market instance inspired from~\cite{6039082}, which considers settings where the supplier cost functions are smooth piecewise quadratic functions. Specifically, for both the fixed and variable demand settings, we consider one supplier with a cost function, which is fixed across the time horizon, given by:
\begin{equation}
c(x) = \begin{cases}
1/32 x^2 + 5/16x & 0\leq x < 1 \\
1/16 x^2 + 1/4x + 1/32 & 1 < x < 2 \\
1/8 x^2 + 9/32  & x \geq 2
\end{cases}
\end{equation}
In the fixed demand setting, we consider a demand of $d = 1.5$. In the variable demand setting, the demand is uniformly drawn at random between the range $[0.1, 4]$ at each period and we average the different regret metrics across 100 instances of demand realizations.

\emph{Instance 2:} We consider a modification of a market instance from~\cite{azizan2020optimal}, which considers a market based on the well-studied markets in~\cite{hogan2003minimum,scarf-example}. For both the fixed and variable demand settings, we consider three suppliers with fixed cost functions $c_1(x_1) = 3/16 x_1^2$, $c_2(x_2) = 2/7 x_2^2$, and $c_3(x_3) = 7/6 x_1^2$ across the time horizon, which correspond to the strongly convex components of the supplier cost functions presented in~\cite{azizan2020optimal}. In the fixed demand setting, we consider a demand of $d = 1$. In the variable demand setting, as in instance 1, the demand is uniformly drawn at random between the range $[0.1, 4]$ at each period and we average the different regret metrics across 100 instances of demand realizations.

We note that our chosen market instances are representative of much of the work on electricity markets that has focused on quadratic or piecewise quadratic supplier cost functions~\citep{TSAOUSOGLOU2022111890,1102918,7796713,4112652}.

\paragraph{Results:}

Figures~\ref{fig:regret-fixed_cost_demand_inst1}-\ref{fig:regret-fixed_cost_variable_demand} depict the unmet demand, cost regret, and payment regret corresponding to Algorithms~\ref{alg:fixed-demand-constant-reg} and~\ref{alg:time-varying-demand-new-sqrt} for the above-defined problem instances for the respective settings of fixed and variable demands. From Figures~\ref{fig:regret-fixed_cost_demand_inst1} and~\ref{fig:regret-fixed_cost_demand}, which correspond to the fixed demand setting for instances 1 and 2, respectively, we observe that all three regret metrics are either a small positive or negative constant, aligning with our obtained regret bound in Theorem~\ref{thm:IdenticalResult}. From the log-log plots of the three regret metrics across time periods depicted in Figures~\ref{fig:regret-fixed_cost_variable_demand_inst1} and~\ref{fig:regret-fixed_cost_variable_demand}, which correspond to the variable demand setting for instances 1 and 2, respectively, we observe that the slopes of the of the respective curves is roughly $0.5$, which aligns with the $O(\sqrt{T})$ regret bound of Algorithm~\ref{alg:time-varying-demand-new-sqrt} obtained in Theorem~\ref{thm:VaryingDemandResult-sqrtT}. Moreover, from the center of Figure~\ref{fig:regret-fixed_cost_variable_demand_inst1}, we observe that the cost regret is negative, which also aligns with the $O(\sqrt{T})$ regret bound of Algorithm~\ref{alg:time-varying-demand-new-sqrt}. Thus, our obtained numerical results validate our developed theoretical bounds for Algorithms~\ref{alg:fixed-demand-constant-reg} and~\ref{alg:time-varying-demand-new-sqrt}. 
\

\begin{figure}
    \centering \begin{subfigure}[t] {0.3\columnwidth}
    \begin{tikzpicture}

\definecolor{darkgray176}{RGB}{176,176,176}
\definecolor{steelblue31119180}{RGB}{31,119,180}

\begin{axis}[
width=2in,
height=2in,
tick align=outside,
tick pos=left,
ylabel={Unmet Demand},
xlabel = {Time (thousands of steps)},
x grid style={darkgray176},
xmin=-0.8, xmax=20.8,
xtick style={color=black},
xticklabels={-1, 0, 5, 10, 15, 20}, y grid style={darkgray176},
ymin=0.5, ymax=6.5,
ytick style={color=black}
]
\addplot [line width = 2pt, steelblue31119180]
table {0.100 3.4999999999999996
0.200 3.4999999999999996
0.300 3.4999999999999996
0.400 3.4999999999999996
0.500 3.4999999999999996
0.750 3.4999999999999996
1.000 3.4999999999999996
2.000 3.4999999999999996
3.000 3.4999999999999996
4.000 3.4999999999999996
5.000 3.4999999999999996
10.000 3.4999999999999996
20.000 3.4999999999999996
};
\end{axis}

\end{tikzpicture} \end{subfigure} \begin{subfigure}[t] {0.3\columnwidth} 
    \begin{tikzpicture}

\definecolor{darkgray176}{RGB}{176,176,176}
\definecolor{steelblue31119180}{RGB}{31,119,180}

\begin{axis}[
width=2in,
height=2in,
tick align=outside,
tick pos=left,
ylabel={Cost Regret},
xlabel = {Time (thousands of steps)},
x grid style={darkgray176},
xmin=-0.8, xmax=20.8,
xtick style={color=black},
xticklabels={-1, 0, 5, 10, 15, 20}, y grid style={darkgray176},
ymin=-2.5, ymax=0,
ytick style={color=black}
]
\addplot [line width = 2pt, steelblue31119180]
table {0.100 -1.2604166666666667
0.200 -1.2604166666666667
0.300 -1.2604166666666667
0.400 -1.2604166666666667
0.500 -1.2604166666666667
0.750 -1.2604166666666667
1.000 -1.2604166666666667
2.000 -1.2604166666666667
3.000 -1.2604166666666667
4.000 -1.2604166666666667
5.000 -1.2604166666666667
10.000 -1.2604166666666667
20.000 -1.2604166666666667
};
\end{axis}

\end{tikzpicture} \end{subfigure} \begin{subfigure}[t] {0.3\columnwidth}
    \begin{tikzpicture}

\definecolor{darkgray176}{RGB}{176,176,176}
\definecolor{steelblue31119180}{RGB}{31,119,180}

\begin{axis}[
    width=2in,
    height=2in,
    tick align=outside,
    tick pos=left,
    ylabel={Payment Regret},
    xlabel={Time (thousands of steps)},
    x grid style={darkgray176},
xmin=-0.8, xmax=20.8,
xtick style={color=black},
xticklabels={-1, 0, 5, 10, 15, 20}, y grid style={darkgray176},
    ymin=-2.75, ymax=-0.75,
    ytick style={color=black},
    yticklabel style={
        /pgf/number format/.cd,
        fixed,
        fixed zerofill,
        precision=2,
        /tikz/.cd
    },
    scaled y ticks=false,
    every axis y label/.append style={at={(ticklabel cs:1.05)}},
    ylabel near ticks
]
\addplot [line width=2pt, steelblue31119180]
table {
0.100 -1.7135416666666663
0.200 -1.7135416666666663
0.300 -1.7135416666666663
0.400 -1.7135416666666663
0.500 -1.7135416666666663
0.750 -1.7135416666666663
1.000 -1.7135416666666663
2.000 -1.7135416666666663
3.000 -1.7135416666666663
4.000 -1.7135416666666663
5.000 -1.7135416666666663
10.000 -1.7135416666666663
20.000 -1.7135416666666663
};
\end{axis}

\end{tikzpicture} \end{subfigure}
\vspace{-10pt}
    \caption{\small \sf Unmet demand, cost regret, and payment regret corresponding to Algorithm~\ref{alg:fixed-demand-constant-reg} for a setting with a fixed demand and one supplier with a fixed cost function described in Instance 1. }
    \label{fig:regret-fixed_cost_demand_inst1}
\end{figure}

\begin{figure}
    \centering \begin{subfigure}[t] {0.3\columnwidth}
    \begin{tikzpicture}
\definecolor{darkgray176}{RGB}{176,176,176}
\definecolor{steelblue31119180}{RGB}{31,119,180}

\begin{axis}[
width=2in,
height=2in,
tick align=outside,
tick pos=left,
ylabel={Log of Unmet Demand},
xlabel = {Log of Time Periods ($\log(T)$)},
x grid style={darkgray176},
xmin=5.5, xmax=11,
xtick style={color=black},
y grid style={darkgray176},
ymin=-2.5, ymax=0.5,
ytick style={color=black},
legend style={
  fill opacity=0.8,
  draw opacity=1,
  text opacity=1,
  at={(0.35,0.25)},
  anchor=north west,
  draw=white!80!black,
  inner xsep=-0.2pt, inner ysep=-2pt,
  font = {\tiny\arraycolsep=2pt}
}
]
\addplot [line width = 2pt, steelblue31119180]
table {5.7037824746562 -2.3462824
5.99146454710798 -2.33835227
6.21460809842219 -1.71206333
6.62007320653036 -1.72278028
6.90775527898214 -1.71088748
7.60090245954208 -1.1304407
8.00636756765025 -0.58809527
8.29404964010203 -0.58575593
8.51719319141624 -0.5869404
9.21034037197618 -0.08618985
9.90348755253613 -0.08454975
10.8197782844103 0.37138437
};
\addlegendentry{Algorithm~\ref{alg:time-varying-demand-new-sqrt}}
\addplot [semithick, red, draw=red, fill=red, mark=*]
table {5.7 -2.1
10.8 0.45
};
\addlegendentry{Slope = 0.5}
\end{axis}

\end{tikzpicture} \end{subfigure} \begin{subfigure}[t] {0.3\columnwidth} 
    \begin{tikzpicture}

\definecolor{darkgray176}{RGB}{176,176,176}
\definecolor{steelblue31119180}{RGB}{31,119,180}

\begin{axis}[
width=2in,
height=2in,
tick align=outside,
tick pos=left,
ylabel={Cost Regret},
xlabel = {Time (thousands of steps)},
x grid style={darkgray176},
xmin=-0.8, xmax=50.8,
xtick style={color=black},
y grid style={darkgray176},
ymin=-300, ymax=0,
ytick style={color=black}
]
\addplot [line width = 2pt, steelblue31119180]
table {0.300 -11.89782260273171
0.400 -12.299256379279567
0.500 -25.32928212330074
0.750 -24.95151571752401
1.000 -26.320461608395803
2.000 -46.11000207117626
3.000 -81.31864476473156
4.000 -82.34173283978687
5.000 -81.31673242771112
10.000 -132.51047314056623
20.000 -133.41933351868366
50.000 -206.49208678999506
};
\end{axis}

\end{tikzpicture} \end{subfigure} \hspace{10pt}
\begin{subfigure}[t] {0.3\columnwidth}
    \begin{tikzpicture}

\definecolor{darkgray176}{RGB}{176,176,176}
\definecolor{steelblue31119180}{RGB}{31,119,180}

\begin{axis}[
width=2in,
height=2in,
tick align=outside,
tick pos=left,
ylabel={Log of Payment Regret},
xlabel = {$\log(T)$},
x grid style={darkgray176},
xmin=5.5, xmax=11,
xtick style={color=black},
y grid style={darkgray176},
ymin=-7, ymax=-4,
ytick style={color=black}
]
\addplot [line width = 2pt, steelblue31119180]
table {5.7037824746562 -6.95145258
5.99146454710798 -6.94352246
6.21460809842219 -6.31723351
6.62007320653036 -6.32795046
6.90775527898214 -6.31605767
7.60090245954208 -5.73561088
8.00636756765025 -5.19326546
8.29404964010203 -5.19092612
8.51719319141624 -5.19211059
9.21034037197618 -4.69136004
9.90348755253613 -4.68971993
10.8197782844103 -4.23378581
};
\addplot [semithick, red, draw=red, fill=red, mark=*]
table {5.7 -6.75
10.8 -4.2
};
\end{axis}

\end{tikzpicture} \end{subfigure}
\vspace{-10pt}
    \caption{\small \sf Unmet demand, cost regret, and payment regret corresponding to Algorithm~\ref{alg:time-varying-demand-new-sqrt} for a setting with a supplier with fixed cost function described in Instance 1 and a variable demand uniformly drawn from $[0.1, 4]$ at each period. }
    \label{fig:regret-fixed_cost_variable_demand_inst1}
\end{figure}

\subsection{Validating Proposition~\ref{prop:time-varying-cost-countereg} for Online Dual Sub-gradient Descent} \label{sec:gd-validation}

In this section, we use numerical experiments to validate our impossibility result (Proposition~\ref{prop:time-varying-cost-countereg}) in the time-varying cost setting. In particular, we show that applying a widely studied primal-dual method that has been shown to achieve good performance (i.e., low regret) for a wide range of problem settings will result in linear regret on at least one of the three regret metrics studied in this work, thereby aligning with Proposition~\ref{prop:time-varying-cost-countereg}.

\paragraph{Online Sub-gradient Descent Algorithm:}

A widely used method in the online learning literature is online dual sub-gradient descent, wherein a dual multiplier representing the price is maintained by the algorithm and this price is updated at each period using dual sub-gradient descent. The dual sub-gradient descent algorithm for our studied problem setting is presented formally in Algorithm~\ref{alg:GradientDescentEnergy}. In particular, Algorithm~\ref{alg:GradientDescentEnergy} proceeds in two phases at each period $t$. In the first phase, suppliers supply the optimal quantity of goods given the set prices for that period. In the second phase, prices for the next period are updated based on the discrepancy between the supply and the demand $d$. To update the price, we use a step-size $\gamma$, typically chosen to be $O(\frac{1}{\sqrt{T}})$, and ensure that the price is non-negative at each period. We note that the price update step follows from performing gradient descent on the $t'$th term of the dual of the central planner's cost minimization Problem~\eqref{eq:supObj2}-\eqref{eq:demand-con}.

\begin{figure}
    \centering \begin{subfigure}[t] {0.3\columnwidth}
    \begin{tikzpicture}

\definecolor{darkgray176}{RGB}{176,176,176}
\definecolor{steelblue31119180}{RGB}{31,119,180}

\begin{axis}[
width=2in,
height=2in,
tick align=outside,
tick pos=left,
ylabel={Unmet Demand},
xlabel = {Time (thousands of steps)},
x grid style={darkgray176},
xmin=-0.8, xmax=20.8,
xtick style={color=black},
xticklabels={-1, 0, 5, 10, 15, 20}, y grid style={darkgray176},
ymin=-0.5, ymax=1.5,
ytick style={color=black}
]
\addplot [line width = 2pt, steelblue31119180]
table {0.100 0.5103108416139558
0.200 0.5103108416139669
0.300 0.510310841613978
0.400 0.5103108416139891
0.500 0.5103108416140002
0.750 0.510310841614028
1.000 0.5103108416140557
2.000 0.5103108416141667
3.000 0.5103108416142778
4.000 0.5103108416143888
5.000 0.5103108416144998
10.000 0.5103108416150549
20.000 0.5103108416161651
};
\end{axis}

\end{tikzpicture} \end{subfigure} \hspace{10pt} 
\begin{subfigure}[t] {0.3\columnwidth} 
    \begin{tikzpicture}

\definecolor{darkgray176}{RGB}{176,176,176}
\definecolor{steelblue31119180}{RGB}{31,119,180}

\begin{axis}[
width=2in,
height=2in,
tick align=outside,
tick pos=left,
ylabel={Cost Regret},
xlabel = {Time (thousands of steps)},
x grid style={darkgray176},
xmin=-0.8, xmax=20.8,
xtick style={color=black},
xticklabels={-1, 0, 5, 10, 15, 20}, y grid style={darkgray176},
ymin=-0.08, ymax=0.03,
ytick style={color=black}
]
\addplot [line width = 2pt, steelblue31119180]
table {0.100 -0.021994985056753118
0.200 -0.021994985056756587
0.300 -0.021994985056760057
0.400 -0.021994985056763526
0.500 -0.021994985056766996
0.750 -0.02199498505677567
1.000 -0.021994985056784343
2.000 -0.021994985056819037
3.000 -0.02199498505685373
4.000 -0.021994985056888426
5.000 -0.02199498505692312
10.000 -0.021994985057096593
20.000 -0.021994985057443538
};
\end{axis}

\end{tikzpicture} \end{subfigure} \begin{subfigure}[t] {0.3\columnwidth}
    \begin{tikzpicture}

\definecolor{darkgray176}{RGB}{176,176,176}
\definecolor{steelblue31119180}{RGB}{31,119,180}

\begin{axis}[
    width=2in,
    height=2in,
    tick align=outside,
    tick pos=left,
    ylabel={Payment Regret},
    xlabel={Time (thousands of steps)},
    x grid style={darkgray176},
xmin=-0.8, xmax=20.8,
xtick style={color=black},
xticklabels={-1, 0, 5, 10, 15, 20}, y grid style={darkgray176},
    ymin=-0.25, ymax=0.75,
    ytick style={color=black},
    yticklabel style={
        /pgf/number format/.cd,
        fixed,
        fixed zerofill,
        precision=2,
        /tikz/.cd
    },
    scaled y ticks=false,
    every axis y label/.append style={at={(ticklabel cs:1.05)}},
    ylabel near ticks
]
\addplot [line width=2pt, steelblue31119180]
table {
0.100 0.25658576796223226
0.200 0.2565857679622267
0.300 0.25658576796222116
0.400 0.2565857679622156
0.500 0.25658576796221005
0.750 0.2565857679621962
1.000 0.2565857679621823
2.000 0.2565857679621268
3.000 0.2565857679620713
4.000 0.25658576796201576
5.000 0.25658576796196025
10.000 0.2565857679616827
20.000 0.2565857679611276
};
\end{axis}

\end{tikzpicture} \end{subfigure}
\vspace{-10pt}
    \caption{\small \sf Unmet demand, cost regret, and payment regret corresponding to Algorithm~\ref{alg:fixed-demand-constant-reg} for a setting with a fixed demand and three suppliers with fixed cost functions described in Instance 2. }
    \label{fig:regret-fixed_cost_demand}
\end{figure}
\begin{figure}
    \centering \begin{subfigure}[t] {0.3\columnwidth}
    \begin{tikzpicture}
\definecolor{darkgray176}{RGB}{176,176,176}
\definecolor{steelblue31119180}{RGB}{31,119,180}

\begin{axis}[
width=2in,
height=2in,
tick align=outside,
tick pos=left,
ylabel={Log of Unmet Demand},
xlabel = {Log of Time Periods ($\log(T)$)},
x grid style={darkgray176},
xmin=5.5, xmax=11,
xtick style={color=black},
y grid style={darkgray176},
ymin=-2.5, ymax=0.5,
ytick style={color=black},
legend style={
  fill opacity=0.8,
  draw opacity=1,
  text opacity=1,
  at={(0.35,0.25)},
  anchor=north west,
  draw=white!80!black,
  inner xsep=-0.2pt, inner ysep=-2pt,
  font = {\tiny\arraycolsep=2pt}
}
]
\addplot [line width = 2pt, steelblue31119180]
table {5.7037824746562 -2.2464284
5.99146454710798 -2.24143233
6.21460809842219 -1.64029674
6.62007320653036 -1.64418369
6.90775527898214 -1.64615977
7.60090245954208 -1.07403269
8.00636756765025 -0.53211072
8.29404964010203 -0.53117528
8.51719319141624 -0.53572088
9.21034037197618 -0.04017009
9.90348755253613 -0.03931497
10.8197782844103 0.41612989
};
\addlegendentry{Algorithm~\ref{alg:time-varying-demand-new-sqrt}}
\addplot [semithick, red, draw=red, fill=red, mark=*]
table {5.7 -2.1
10.8 0.45
};
\addlegendentry{Slope = 0.5}
\end{axis}

\end{tikzpicture}
 \end{subfigure} \begin{subfigure}[t] {0.3\columnwidth} 
    \begin{tikzpicture}

\definecolor{darkgray176}{RGB}{176,176,176}
\definecolor{steelblue31119180}{RGB}{31,119,180}

\begin{axis}[
width=2in,
height=2in,
tick align=outside,
tick pos=left,
ylabel={Log of Cost Regret},
xlabel = {$\log(T)$},
x grid style={darkgray176},
xmin=5.5, xmax=11,
xtick style={color=black},
y grid style={darkgray176},
ymin=0.25, ymax=3.5,
ytick style={color=black}
]
\addplot [line width = 2pt, steelblue31119180]
table {5.7037824746562 0.46454021
5.99146454710798 0.69456467
6.21460809842219 1.2156263
6.62007320653036 1.32339675
6.90775527898214 1.20542171
7.60090245954208 1.87755725
8.00636756765025 2.47788053
8.29404964010203 2.46636923
8.51719319141624 2.44288202
9.21034037197618 3.01343841
9.90348755253613 3.02394327
10.8197782844103 3.45316406
};
\addplot [semithick, red, draw=red, fill=red, mark=*]
table {5.7 0.85
10.8 3.4
};
\end{axis}

\end{tikzpicture} \end{subfigure} \hspace{10pt}
\begin{subfigure}[t] {0.3\columnwidth}
    \begin{tikzpicture}

\definecolor{darkgray176}{RGB}{176,176,176}
\definecolor{steelblue31119180}{RGB}{31,119,180}

\begin{axis}[
width=2in,
height=2in,
tick align=outside,
tick pos=left,
ylabel={Log of Payment Regret},
xlabel = {$\log(T)$},
x grid style={darkgray176},
xmin=5.5, xmax=11,
xtick style={color=black},
y grid style={darkgray176},
ymin=-7, ymax=-4,
ytick style={color=black}
]
\addplot [line width = 2pt, steelblue31119180]
table {5.7037824746562 -6.85159859
5.99146454710798 -6.84660252
6.21460809842219 -6.24546693
6.62007320653036 -6.24935388
6.90775527898214 -6.25132996
7.60090245954208 -5.67920288
8.00636756765025 -5.1372809
8.29404964010203 -5.13634546
8.51719319141624 -5.14089106
9.21034037197618 -4.64534027
9.90348755253613 -4.64448515
10.8197782844103 -4.1890403
};
\addplot [semithick, red, draw=red, fill=red, mark=*]
table {5.7 -6.75
10.8 -4.2
};
\end{axis}

\end{tikzpicture} \end{subfigure}
\vspace{-10pt}
    \caption{\small \sf Unmet demand, cost regret, and payment regret corresponding to Algorithm~\ref{alg:time-varying-demand-new-sqrt} for a setting with three suppliers with fixed cost functions described in Instance 2 and a variable demand uniformly drawn from $[0.1, 4]$ at each period. }
    \label{fig:regret-fixed_cost_variable_demand}
\end{figure}

\vspace{-10pt}
\begin{algorithm} 
\SetAlgoLined
\SetKwInOut{Input}{Input}\SetKwInOut{Output}{Output}
Initialize $p_1 = 0$\;
\For{$t = 1, \ldots, T$}{
 \textbf{Phase I: Suppliers Produce Profit Maximizing Bundle} \\
 $\Tilde{x}_{it} = \argmax_{x_{it} \geq 0} p_t x_{it} - c_{it}(x_{it})$ for all suppliers $i \in [n]$ \;
 $ \Tilde{X}_t = \sum_{i = 1}^n \Tilde{x}_{it}$ \tcp*{Observe Total Production}
 \textbf{Phase II: Price Vector Update} \\
 $p_{(t+1)} \leftarrow (p_{(t)} - \gamma (\Tilde{X}_t - d))_{+}$ \;
 }
\caption{Dual Sub-gradient Descent}
\label{alg:GradientDescentEnergy}
\end{algorithm}
\vspace{-10pt}

\
\paragraph{Experiment Setup:} We consider a setting with one supplier and a fixed customer demand of $d = 1$ at each period. Furthermore, we let the supplier's cost function be drawn from a probability distribution such that $c_1(x_1) = \frac{1}{6} x_1^2$ with probability $0.5$ and $c_1(x_1) = \frac{1}{12} x_1^2 + \frac{1}{4} x_1$ with probability $0.5$ at each period.

\paragraph{Results:} 

Figure~\ref{fig:regret-gradient-descent} depicts the unmet demand, cost regret, and payment regret corresponding to Algorithm~\ref{alg:GradientDescentEnergy} for the above-defined problem instance when the step-size $\gamma = \frac{1}{\sqrt{T}}$. In particular, we observe that while the payment regret is negative, the cost regret and unmet demand corresponding to Algorithm~\ref{alg:GradientDescentEnergy} increase linearly in the time horizon. Such a result validates Proposition~\ref{prop:time-varying-cost-countereg}, which states that no algorithm can achieve sub-linear regret on all three regret metrics in the setting when the costs vary over time; thus, in particular, Algorithm~\ref{alg:GradientDescentEnergy} cannot achieve sub-linear regret on all three regret metrics in the time-varying cost setting. We note that this result is in contrast to prior online learning literature~\citep{balseiro2022best}, where, in the setting when the cost functions are drawn i.i.d. from some distribution, online dual sub-gradient descent methods achieve sub-linear regret. 
\

\begin{figure}
    \centering \hspace{-5pt}
\begin{subfigure}[t] {0.3\columnwidth}
    \begin{tikzpicture}

\definecolor{darkgray176}{RGB}{176,176,176}
\definecolor{steelblue31119180}{RGB}{31,119,180}

\begin{axis}[
width=2in,
height=2in,
tick align=outside,
tick pos=left,
ylabel={Unmet Demand},
xlabel = {Time (thousands of steps)},
x grid style={darkgray176},
xmin=-0.8, xmax=20.8,
xtick style={color=black},
xticklabels={-1, 0, 5, 10, 15, 20}, y grid style={darkgray176},
ymin=-200, ymax=2000,
ytick style={color=black}
]
\addplot [line width = 2pt, steelblue31119180]
table {0.100 10.596610385386214
0.200 20.350717176475644
0.300 31.430366471573745
0.400 40.743322728756596
0.500 48.70229829459443
0.750 72.68937934016162
1.000 92.76291085182733
2.000 183.14479396876416
3.000 269.93116567631216
4.000 358.7638198673388
5.000 449.091692498243
10.000 885.3757420852353
20.000 1753.6527396464521
};
\end{axis}

\end{tikzpicture} \end{subfigure} \hspace{5pt} 
\begin{subfigure}[t] {0.3\columnwidth} 
    \begin{tikzpicture}

\definecolor{darkgray176}{RGB}{176,176,176}
\definecolor{steelblue31119180}{RGB}{31,119,180}

\begin{axis}[
width=2in,
height=2in,
tick align=outside,
tick pos=left,
ylabel={Cost Regret},
xlabel = {Time (thousands of steps)},
x grid style={darkgray176},
xmin=-0.8, xmax=20.8,
xtick style={color=black},
xticklabels={-1, 0, 5, 10, 15, 20}, y grid style={darkgray176},
ymin=-10, ymax=80,
ytick style={color=black}
]
\addplot [line width = 2pt, steelblue31119180]
table {0.100 0.13439331502886498
0.200 0.3470852030661555
0.300 0.670978241769131
0.400 0.9757956578927024
0.500 1.1975388098090454
0.750 1.952707577668813
1.000 2.7560331129941877
2.000 5.6547431697193975
3.000 8.81476954050949
4.000 12.20252723544952
5.000 15.45950814735463
10.000 31.63905654850833
20.000 64.32915589045402
};
\end{axis}

\end{tikzpicture} \end{subfigure} \begin{subfigure}[t] {0.3\columnwidth}
    \begin{tikzpicture}

\definecolor{darkgray176}{RGB}{176,176,176}
\definecolor{steelblue31119180}{RGB}{31,119,180}

\begin{axis}[
width=2in,
height=2in,
tick align=outside,
tick pos=left,
ylabel={Payment Regret},
xlabel = {Time (thousands of steps)},
x grid style={darkgray176},
xmin=-0.8, xmax=20.8,
xtick style={color=black},
xticklabels={-1, 0, 5, 10, 15, 20}, y grid style={darkgray176},
ymin=-4000, ymax=300,
ytick style={color=black}
]
\addplot [line width = 2pt, steelblue31119180]
table {0.100 -17.773531830606654
0.200 -35.01482826064632
0.300 -54.16351118718205
0.400 -71.47626743540148
0.500 -88.38597320532239
0.750 -135.7458835390493
1.000 -175.3440349821957
2.000 -355.5094879530044
3.000 -534.1978205019286
4.000 -702.9626592720701
5.000 -887.2171193407386
10.000 -1771.3322412651944
20.000 -3527.4779828547344
};
\end{axis}

\end{tikzpicture} \end{subfigure}
\vspace{-10pt}
    \caption{\small \sf Unmet demand, cost regret, and payment regret of Algorithm~\ref{alg:GradientDescentEnergy} for a setting with a fixed demand of $d = 1.1$ and one supplier with a cost function that is drawn from a probability distribution such that $c_1(x_1) = 2 x_1^2$ with probability $0.5$ and $c_1(x_1) = 4 x_1^2$ with probability $0.5$ at each period. } \label{fig:regret-gradient-descent}
\end{figure} 
\
\section{Discussion of Model Extensions and Future Work} \label{sec:discussion}

There are several natural extensions to the model studied in this work. In this section, we highlight two key directions that incorporate important practical considerations and offer potential avenues for building on the current framework. Section~\ref{subsubsec:strategic-suppliers} explores a setting in which suppliers are forward-looking, wherein they may strategically misreport production at particular periods to achieve better outcomes in subsequent periods. Next, Section~\ref{subsubsec:relaxed-metric-gd} examines relaxations of the unmet demand regret metric, e.g., incorporating storage, and demonstrates that even under some of these relaxed metrics, it is not possible to circumvent the impossibility results presented in this work. Finally, Section~\ref{subsec:other-future-directions} presents some additional avenues for future work.

\subsection{Strategic Suppliers Misreporting Productions} \label{subsubsec:strategic-suppliers}

In this work, we studied a revealed preference setting in which suppliers, given posted prices, produce the optimal quantity of the commodity at each period. A natural extension of this framework involves relaxing the assumption of myopic supplier behavior and instead considering forward-looking suppliers who may act strategically. Specifically, rather than producing the optimal quantity in each period, such suppliers may misreport production or deliberately deviate from optimal actions in particular periods to influence future prices and thereby achieve improved outcomes in subsequent periods. In the following, we present an illustrative example demonstrating that even when customer demand and supplier cost functions remain fixed across time, our proposed algorithms can be \emph{gamed}, i.e., a strategic, forward-looking supplier can manipulate reported production in one period to induce a more favorable price in a later period, thereby distorting the learning of equilibrium prices.

\begin{example} [Gains from Supplier Misreporting Production] \label{eg:gains-supplier-production}
Consider a setting with $T = 2$ periods and a single supplier with a fixed cost function $c(x) = x^2$ and a fixed demand of $d = 0.2$. We will now show that the supplier can generate a higher profit by misreporting its optimal production in the first period under Algorithm~\ref{alg:fixed-demand-constant-reg}.

To see this, we first compute the supplier's profit, if they reveal their optimal productions at each step truthfully. Note that the price at the first period under Algorithm~\ref{alg:fixed-demand-constant-reg} is $0.5$ and at this price, the optimal production for the supplier is $x^*(0.5) = 0.5/2 = 0.25$, resulting in a first-period profit of $0.5 \times 0.25 - 0.25^2 = 1/16$. Next, since $x^*(0.5) > d$, the price at the second period under Algorithm~\ref{alg:fixed-demand-constant-reg} is $0.25$. Under this price, the optimal production for the supplier is $x^*(0.25) = 0.25/2 = 1/8$, resulting in a second-period profit of $0.25 \times 0.125 - 0.125^2 = 1/64$. Thus, the total profit across the two periods if the supplier reveals its optimal production at each step truthfully is given by $1/16 + 1/64 = 5/64$.

Next, we consider a forward-looking supplier and compute the profit of the supplier if they misreport their production in the first period to get a more favorable price in the second period. Specifically, suppose that the supplier reports a production of $x_{1} = 0.19$ in the first period and $x_{2} = 3/8$ in the second period. In this case, note that the profit in the first period is given by $0.5 \times 0.19 - 0.19^2 = 0.0589$. Since $x_1 < d$, it follows that the price at the second period under Algorithm~\ref{alg:fixed-demand-constant-reg} is $0.75$. Under this price, the optimal production for the supplier is $x^*(0.75) = 0.75/2 = 3/8 = x_2$, resulting in a second-period profit of $0.75 \times 3/8 - (3/8)^2 = 9/64$. Thus, the total profit across the two periods if the supplier misreports its production in the first period is given by $0.0589 + 9/64$, which is greater than the supplier's profit of $5/64$ under truthful reporting of its per-period optimal production quantities. \hfill$\square$
\end{example}

Example~\ref{eg:gains-supplier-production} underscores the challenges that arise in the presence of forward-looking strategic agents, particularly with respect to the incorrect learning of equilibrium prices in our setting. This issue reflects a broader difficulty in the online learning literature, wherein designing mechanisms that guarantee sublinear regret becomes significantly more challenging when agents behave strategically by acting suboptimally in one period to obtain more favorable outcomes in future periods (see, e.g.,~\cite{pmlr-v119-freeman20a}). As such, Example~\ref{eg:gains-supplier-production} illustrates the critical need for further investigation into models that account for strategic forward-looking supplier behavior in our studied setting and the development of algorithms that are robust to such strategic manipulation, which we believe are important directions for future research.

\subsection{Relaxing the Unmet Demand Metric} \label{subsubsec:relaxed-metric-gd}

In this work, we studied an unmet demand metric that requires demand to be approximately satisfied in each period, such that overproduction in one period cannot offset unmet demand in another. In contrast, the cost and payment regret metrics are evaluated in aggregate over time, allowing for long-run satisfaction (see Section~\ref{sec:perf-measures}). To that end, a natural extension of our framework is to consider relaxations of the unmet demand metric that are representative of particular practical contexts and better understand if they help mitigate some of the impossibility results established in this work for the settings they are relevant for. 
One such relaxation is to evaluate unmet demand only in aggregate, defined as $\Hat{U}_T(\ppi) = \left[ \sum_{t = 1}^T \left( d_t - \sum_{i = 1}^n x_{it}^*(p_t) \right) \right]_+$, analogous to the cumulative treatment of cost and payment regret. 

However, even under this relaxed metric, the impossibility of achieving sublinear regret for non-strongly convex cost functions persists, which can be shown via a straightforward adaptation of the proof of Proposition~\ref{prop:impossibility-strongconvex}. Furthermore, the impossibility result for the time-varying cost setting (Proposition~\ref{prop:time-varying-cost-countereg}) also continues to hold under the relaxed unmet demand metric, as is elucidated through the following proposition.

\begin{proposition} [Impossibility of Sub-linear Regret under Relaxed Unmet Demand Metric] \label{prop:time-varying-cost-countereg-2}
There exists an instance with fixed time-invariant demand and a single supplier whose cost functions are drawn i.i.d. from some (potentially known) distribution such that no online algorithm can achieve sub-linear regret on payment regret, cost regret, and the relaxed unmet demand metric $\Hat{U}_T(\ppi) = \left[ \sum_{t = 1}^T \left( d_t - \sum_{i = 1}^n x_{it}^*(p_t) \right) \right]_+$.
\end{proposition}

To establish Proposition~\ref{prop:time-varying-cost-countereg-2}, we construct a different counterexample than the one in Proposition~\ref{prop:time-varying-cost-countereg} and we refer to Appendix~\ref{apdx:pf-stronger-example} for its proof.

Beyond the aforementioned aggregate relaxation of the unmet demand metric, other intermediate formulations between the strict per-period notion studied in this work and the relaxed long-run version described above are also plausible. For instance, the unmet demand metric can be modified to incorporate storage, as in many real-world electricity markets, allowing excess supply in earlier periods to meet future demand. While incorporating storage in the unmet demand metric is unlikely to overcome the impossibility results presented in this work, we view the exploration of relaxed unmet demand metrics and benchmarks, which reflect operational realities while enabling improved regret guarantees, as a promising and valuable direction for future research.

\subsection{Other Future Directions} \label{subsec:other-future-directions}

Beyond the model extensions discussed in Sections~\ref{subsubsec:strategic-suppliers} and~\ref{subsubsec:relaxed-metric-gd}, there are several other natural directions for future research.   
For instance, there is a scope to generalize the model to the setting when suppliers' cost functions are non-convex, in which case the operator may need different pricing strategies for each supplier~\citep{azizan2020optimal}.
Another valuable extension would be the study of networked markets, where the suppliers correspond to the nodes and are connected through potentially capacitated edges.

\ 
\section{Conclusion}

In this work, we studied the problem of setting equilibrium prices to satisfy the customer demand for a commodity in markets where the cost functions of suppliers are unknown to the market operator. Since centralized optimization approaches to compute equilibrium prices are typically not conducive in this incomplete information setting, we studied the problem of learning equilibrium prices online under several informational settings regarding the time-varying nature of the customer demands and supplier cost functions. We first considered the setting when suppliers' cost functions are fixed over the $T$ periods and developed posted-price algorithms with regret guarantees of $O(1)$ (and $O(\sqrt{T})$) when the customer demand is fixed (or can vary across the periods in a continuous interval) for strongly convex cost functions. Next, when suppliers' cost functions are time-varying, we showed that no online algorithm achieves sub-linear regret on all three regret metrics when suppliers' cost functions are sampled i.i.d. from a distribution. Thus, we studied an augmented contextual bandit setting where the operator has access to hints (contexts) on how the cost functions change over time and developed a posted-price algorithm that with sub-linear regret on all three regret metrics in this setting. Finally, we presented numerical experiments that validate our obtained regret guarantees and discussed various extensions to the model proposed in this work.

\bibliographystyle{plainnat}
\bibliography{main}

\appendix

\section{Proof of Technical Lemmas}
\label{sec:lipschitz-lemma}

\begin{proof}[Proof of Lemma~\ref{lem:ProductionLipschitz}]
Fix a period $t \in [T]$. Then, by computing the first-order optimally condition of Problem~\eqref{eq:supObj} for each supplier $i$, we have that 
\[ p = c_{it}'(x_{it}^*(p)). \]

Then, due to the strong convexity of $c_{it}$'s, for any two prices $p_1, p_2$ where $x^*_{it}(p_1) = x_1, x^*_{it}(p_2) = x_2$, we have
\begin{alignat*}{2}
    &  \mu (x_1 - x_2)^2 &&\le (c_{it}'(x_1) - c_{it}'(x_2))(x_1 - x_2)\\
    \implies\quad& \mu |x_1 - x_2| &&\le |c_{it}'(x_1) - c_{it}'(x_2)| \\
    \implies\quad& |x^*_{it}(p_1) - x^*_{it}(p_2)| && \le \frac{1}{\mu} |p_1 - p_2| .
\end{alignat*}

Hence, at each period $t$, the optimal supplier production $x^*_{it}$ is $(1/\mu)$-Lipschitz in the prices.
\end{proof}

\vspace{1em}

\begin{proof}[Proof of Lemma~\ref{lem:lipschitzPricesRegret}]
We have the following relation for the payment regret:
\begin{align*}
    p_t \sum_{i = 1}^n x_{it}(p) - p^*_t d_t &\stackrel{(a)}{=} (p_t - p^*_t) \sum_{i = 1}^n x_{it}(p_t) + p^*_t \left(\sum_{i = 1}^n x_{it}(p_t) - \sum_{i = 1}^n x_{it}(p^*_t) \right), \\
    &\stackrel{(b)}{\leq} (p_t - p^*_t) \sum_{i = 1}^n x_{it}(p_t) + p^*_t \left(\sum_{i = 1}^n | x_{it}(p_t) - x_{it}(p^*_t) | \right), \\
    &\stackrel{(c)}{\leq} (p_t - p^*_t) n \Bar{x} + n L_0 (p_t - p^*_t) \le (n \Bar{x} + n L_0) |p_t - p^*_t|,
\end{align*}
where (a) follows from adding and subtracting $p^*_t \sum_{i = 1}^n x_{it}(p_t)$, (b) follows as $x_{it}(p_t) - x_{it}(p^*_t) \leq |x_{it}(p_t) - x_{it}(p^*_t)|$ for all $t \in [T]$, and (c) follows by the boundedness of the production $x_{it}(p_t) \leq x_{it}(1) \leq \Bar{x}$ for all prices $p_t \in [0, 1]$ and the Lipschitzness of the production in the prices (see Lemma~\ref{lem:ProductionLipschitz}).

To complete this bound, we show that the maximum production $\bar{x}$ of each supplier is at most $2/\mu$.
Because the supplier's cost function $c_{it}$ is $\mu$-strongly convex, for any posted price $p$, we have 
\[c_{it}(0) \ge c_{it}(x_{it}(p)) + c'_{it}(x_{it}(p))(0 - x_{it}(p)) + \frac{\mu}{2} x_{it}(p)^2.\]
Due to first-order optimality condition, $p  = c'_{it}(x_{it}(p))$.
Therefore,
\[1 \ge p = c'_{it}(x_{it}(p)) \ge \frac{1}{x_{it}(p)}\left(c_{it}(x_{it}(p)) - c_{it}(0) + \frac{\mu}{2} x_{it}(p)^2\right) \ge \frac{\mu}{2} x_{it}(p).\]
Thus, $\bar{x} = \max_{i, t} \max_{p} x_{it}(p) \le 2/\mu.$
We conclude that
\[p_t \sum_{i = 1}^n x_{it}(p) - p^*_t d_t \le n (2/\mu + L_0) |p_t - p^*_t|.\]

Analogously, we have the following relation for the cost regret for a constant $L_0$:
\begin{align*}
    \sum_{i = 1}^n c_{it}(\mathbf{x}_{it}(p_t)) - \sum_{i = 1}^n c_{it}(\mathbf{x}_{it}(p^*_t)) &\stackrel{(a)}{\leq} \sum_{i = 1}^n \nabla c_{it}(\mathbf{x}_{it}(p_t)) (\mathbf{x}_{it}(p_t) - \mathbf{x}_{it}(p^*_t)), \\
    &\stackrel{(b)}{\leq} \sum_{i = 1}^n \nabla c_{it}(\mathbf{x}_{it}(p_t))  L_0 |p_t-p^*|, \\
    &\stackrel{(c)}{\leq} n L_0 |p_t-p^*|.
\end{align*}
where (a) follows by the convexity of the cost functions, (b) follows as the cost functions are monotonically increasing and by the Lipshitzness of the production in the prices (see Lemma~\ref{lem:ProductionLipschitz}), and (c) follows from first-order optimality condition that $c'_{it}(\mathbf{x}_{it}(p_t)) = p_t \le 1$.

Now, summing the above inequalities over $t = 1, \dots, T$ yields the desired results.
\end{proof}

\vspace{1em}

\begin{proof}[Proof of Lemma~\ref{lem:PriceLipschitz}]
Defining the conjugate function $f_{it}^*(p) = \min_{x_{it} \geq 0} \{ c_{it}(x_{it}) - p x_{it} \}$, we first formulate the following dual of Problem~\eqref{eq:supObj2}-\eqref{eq:demand-con}:
\begin{maxi!}|s|[2]                   {p}                               {g(p, d_t) := p d_t + \sum_{i = 1}^n f_{it}^*(p). \label{eq:supObj2Dual}}   {\label{eq:minCostDual}}             {}                                \end{maxi!}
Recall that the cost functions $c_{it}$ are $\ell$-Lipschitz smooth.
Then, from the properties of the conjugate function, the dual function is $1/\ell$-strongly concave (see~\cite{zhou2018fenchel}).
By the strong concavity of the dual function, it follows that
\begin{align}
    \frac{\mu}{2} (p^*(d_1) - p^*(d_2))^2 
    &\leq -g(p^*(d_1), d_2) - (-g(p^*(d_2), d_2)) + \nabla (-g(p^*(d_2), d_2)) (p^*(d_2) - p^*(d_1)) \nonumber \\
    &= -g(p^*(d_1), d_2) + g(p^*(d_2), d_2), \label{eq:strong-convexity-relation}
\end{align}
where the equality follows as $\nabla g(p^*(d_2), d_2) = d_2 - \sum_{i = 1}^n x_{it}^*(p^*(d_2)) = 0$.

Next, we observe that
\begin{align*}
    - g(p^*(d_1), d_2) + g(p^*(d_2), d_2) &= [- g(p^*(d_1), d_2) + g(p^*(d_1), d_1)] - [- g(p^*(d_2), d_2) + g(p^*(d_2), d_1)] \\
    &\hspace{10em} + [- g(p^*(d_1), d_1) + g(p^*(d_2), d_1)], \\
    &\stackrel{(a)}{\leq} [-g(p^*(d_1), d_2) + g(p^*(d_1), d_1)] - [-g(p^*(d_2), d_2) + g(p^*(d_2), d_1)], \\
    &\stackrel{(b)}{=} \left[ -p^*(d_1) d_2 - \sum_{i = 1}^n f_{it}^*(p^*(d_1)) + p^*(d_1) d_1 + \sum_{i = 1}^n f_{it}^*(p^*(d_1)) \right] \\
    &\hspace{2em}- \left[- p^*(d_2) d_2 - \sum_{i = 1}^n f_{it}^*(p^*(d_2)) + p^*(d_2) d_1 + \sum_{i = 1}^n f_{it}^*(p^*(d_2)) \right] \\
    &= (p^*(d_2) - p^*(d_1)) (d_2 - d_1).
\end{align*}
where (a) follows as $g(p^*(d_1), d_1) \geq g(p^*(d_2), d_1)$ by the optimality of the dual function $g$ at the optimal price $p^*(d_1)$ for the demand $d_1$, and (b) follows by the definition of $g$.

From the above inequality and the strong concavity relation in Equation~\eqref{eq:strong-convexity-relation}, we obtain that
\begin{align*}
    \frac{\mu}{2} (p^*(d_1) - p^*(d_2))^2 \leq (p^*(d_1) - p^*(d_2)) (d_2 - d_1) \leq |p^*(d_1) - p^*(d_2)| |d_2 - d_1|,
\end{align*}
which in turn implies our desired Lipschitz condition that
\begin{align*}
    |p^*(d_1) - p^*(d_2)| \leq 2\ell |d_1 - d_2|,
\end{align*}
which establishes our claim for the Lipschitz constant $L_1 = 2\ell$.
\end{proof} 

\section{Proof of Theorem~\ref{thm:VaryingDemandResult-sqrtT}}
\label{sec:vary-demand-pf}

\begin{claim}
Let $\tau = t_1 - t_0$.
Under the choice of $\gamma = 1/\sqrt{\tau}$, the unmet demand, cost regret, and payment regret of Algorithm~\ref{alg:time-varying-demand-helper-sqrt} are $O(\sqrt{\tau})$ if the cost functions of the suppliers are strongly convex.
\end{claim}

\begin{proof}
Denote $p^*(d)$ as the equilibrium price with demand $d$.
Motivated by the unified regret metric \eqref{equ:unified-metric-price}, we try to relate each of the three performance metrics with the quantity
\[\sum_{t = t_0}^{t_1-1} |p_t  - p^*(r_{k_t})|.\]

\begin{itemize}
\item For unmet demand, we have:
\begin{align*}
    U_T(\ppi) 
    &= \sum_{t = t_0}^{t_1-1} \left( d_t - \sum_{i = 1}^n x_{it}^*(p_t) \right)_+ \\
    &\le \sum_{t = t_0}^{t_1-1} \gamma + \left(r_{k_t} - \sum_{i = 1}^n x_{it}^*(p_t) \right)_+ \\
    &\le \sqrt{\tau} + L_0 \sum_{t = t_0}^{t_1-1} |p_t  - p^*(r_{k_t})|,
\end{align*}
where the second line follows from $d_t \le r_{k_t+1} = r_{k_t} + \gamma$, and third line follows from Lemma \ref{lem:ProductionLipschitz}.

\item For cost regret, we have:
\begin{align*}
    C_T(\ppi) 
    &= \sum_{t = t_0}^{t_1-1} \sum_{i = 1}^n \left(c_{it}(x_{it}^*(p_t)) - c_{it}(x_{it}^*(p^*(d_t))\right) \\
    &\le \sum_{t = t_0}^{t_1-1} \sum_{i = 1}^n \left(c_{it}(x_{it}^*(p_t)) - c_{it}(x_{it}^*(p^*(r_{k_t}))\right) \\
    &\le L_1 \sum_{t = t_0}^{t_1-1} |p_t  - p^*(r_{k_t})|,
\end{align*}
where for the first line we note that the suppliers' productions are monotonic in price and the cost function is increasing.
The second line follows from Lemma~\ref{lem:lipschitzPricesRegret}.

\item For payment regret, we have:
\begin{align*}
    P_T(\ppi) 
    &= \sum_{t = t_0}^{t_1-1} \sum_{i = 1}^n \left(p_t x_{it}^*(p_t) -  p^*(d_t) x_{it}^*(p^*(d_t)) \right) \\
    &\le \sum_{t = t_0}^{t_1-1} \sum_{i = 1}^n \left(p_t x_{it}^*(p_t) -  p^*(r_{k_t}) x_{it}^*(p^*(r_{k_t})) \right) \\
    &\le L_2 \sum_{t = t_0}^{t_1-1} |p_t  - p^*(r_{k_t})|,
\end{align*}
where for the second line we note that the suppliers' productions are monotonic in price and the third line follows from Lemma~\ref{lem:lipschitzPricesRegret}.
\end{itemize}

Now, it suffices to show that
\[\sum_{t = t_0}^{t_1-1} |p_t  - p^*(r_{k_t})| \in O(\sqrt{\tau}).\]
Let $K = (\overline{d} - \underline{d})/\gamma$, we group the sum according to the demand interval:
\begin{align*}
    \sum_{t = t_0}^{t_1-1} |p_t  - p^*(r_{k_t})|
    &= \sum_{k = 1}^K \sum_{t : d_t \in I_k} |p_t  - p^*(r_{k})| \\
    &\le \sum_{k = 1}^K O(1) \in O(K) = O(\sqrt{\tau}),
\end{align*}
where we get the second line by applying Theorem~\ref{thm:IdenticalResult} to each $k$.
Therefore, we conclude that all three performance metrics are at most $O(\sqrt{\tau})$ between time $t_0$ and $t_1$.
\end{proof}

Now, we complete the proof of Theorem~\ref{thm:VaryingDemandResult-sqrtT}.
As we just showed in the previous claim, there exists a universal constant $c > 0$ so that each invocation of Algorithm~\ref{alg:time-varying-demand-helper-sqrt} incurs regret up to $c \cdot \sqrt{t_1 - t_0}$ on all three performance metrics.
Given a total time horizon of length $T$, we then sum over all episodes $m = 0, 1, \dots \lceil \log_2 T \rceil$, so that the regret on each metric is at most:
\begin{align*}
    \sum_{m=0}^{\lceil \log_2 T \rceil} c \cdot \sqrt{2^{m+1} - 2^{m}}
    &\le \sum_{m=0}^{\lceil \log_2 T \rceil} c \cdot 2^{m/2} \\
    &\le c \cdot \int_{0}^{\log_2 T + 2} 2^{y/2} \, \mathrm{d}y \\
    &= \frac{2c}{\ln 2} (2\sqrt{T} - 1) \in O(\sqrt{T})
\end{align*}

Therefore, Algorithm~\ref{alg:time-varying-demand-new-sqrt} achieves $O(\sqrt{T})$ regret on unmet demand, payment regret and cost regret, respectively.

\section{Impossibility Results in the Setting of Time-Varying Costs}
\subsection{Proof of Proposition~\ref{prop:time-varying-cost-countereg}} \label{sec:pf-prop-countereg}

We consider a setting with a fixed demand of $d=1$ at every period and a single supplier whose cost functions at each period are drawn from a distribution such that at each period $t$, its cost function could be either $c_1(x) = \frac{1}{8}x^2$ or $c_2(x) =\frac{1}{16} x^2$, each with probability $0.5$. We suppose that the market operator has knowledge of the distribution from which the supplier's cost function is sampled i.i.d. but does not know the outcome of the random draw at any period and show that any pricing strategy adopted by the operator must incur a linear regret one at least one of the three performance measures for this instance. 

To prove this claim, we first define the \emph{total regret} as the sum of the unmet demand, payment regret, and cost regret and note that if the total regret is linear in the number of periods $T$, then at least one of the three performance measures must be linear in $T$. To analyze the total regret, we first analyze each of the performance measures for a given price $p$ for both cost functions.

\begin{enumerate}
    \item For the first cost function $c_1(x)$, the optimal production level given a price $p$ is $x^*(p) = 4p$, so the equilibrium price is $p^* = \frac{1}{4}$.
    \begin{itemize}
        \item The payment regret at price $p$ is $p(4p) - 1/4 = 4p^2 - 1/4$.
        \item The cost regret at price $p$ is $\frac{1}{8}(4p)^2 - 1/8 = 2p^2 - 1/8$.
        \item The unmet demand is $1-4p$ if $p < 1/4$ and 0 otherwise.
    \end{itemize}
    \item For the second cost function $c_2(x)$, the optimal production level given a price $p$ is $x^*(p) = 8p$, so the equilibrium price is $p^* = \frac{1}{8}$.
    \begin{itemize}
        \item The payment regret at price $p$ is $p(8p) - 1/8 = 8p^2 - 1/8$.
        \item The cost regret at price $p$ is $\frac{1}{16}(8p)^2 - 1/16 = 4p^2 - 1/16$.
        \item The unmet demand is $1-8p$ if $p < 1/8$ and 0 otherwise.
    \end{itemize}
\end{enumerate}
Then, the expected total regret at each period $t$ is as follows:
\begin{itemize}
    \item If $p < 1/8$: expected total regret is $\frac{1}{2}(18p^2 - 9/16 + (1-4p) + (1-8p)) = 9p^2 - 6p + 23/32$.
    \item If $1/8 \le p \le 1/4$: expected total regret is $\frac{1}{2}(18p^2 - 9/16 + (1-4p)) = 9p^2 - 2p + 7/32$.
    \item If $p > 1/4$: expected total regret is $\frac{1}{2}(18p^2 - 9/16) = 9p^2 - 9/32$.
\end{itemize}
From the above obtained relations, we can derive that the expected total regret at any period $t$ is at least $7/64$, which is attained when $p = 1/8$.
It thus follows that, regardless of the pricing strategy adopted by the market operator, the total expected regret is at least $\frac{7}{64} T$.
Hence, either the unmet demand, payment regret, or cost regret are not sublinear, which establishes our claim.

\subsection{Proof of Proposition~\ref{prop:time-varying-cost-countereg-2}} \label{apdx:pf-stronger-example}

We consider a setting with a fixed demand of $d=1$ at every period and a single supplier whose cost functions at each period are drawn from a distribution such that at each period $t$, its cost function could be either $c_1(x) = \frac{1}{6}x^2$ or $c_2(x) =\frac{1}{12} x^2 + \frac{1}{4}x$, each with probability $0.5$. We suppose that the market operator has knowledge of the distribution from which the supplier's cost function is sampled i.i.d. but does not know the outcome of the random draw at any period and show that any pricing strategy adopted by the operator must incur a linear regret on at least one of the three performance metrics for this instance. 

Similar to the previous section, we define the \emph{total regret} as the sum of the relaxed unmet demand, payment regret, and cost regret and note that if the total regret is linear in the number of periods $T$, then at least one of the three performance metrics must be linear in $T$. To analyze the total regret, we first analyze each of the performance metrics for a given price $p$ for both cost functions.

\begin{enumerate}
    \item For the first cost function $c_1(x)$, the optimal production level given a price $p$ is $x^*(p) = 3p$, so the equilibrium price is $p^* = \frac{1}{3}$.
    \begin{itemize}
        \item The payment regret at price $p$ is $p(3p) - 1/3 = 3p^2 - 1/3$.
        \item The cost regret at price $p$ is $\frac{1}{6}(3p)^2 - 1/6 = \frac{3}{2}p^2 - 1/6$.
        \item The relaxed unmet demand is $1-3p$.
    \end{itemize}
    \item For the second cost function $c_2(x)$, the optimal production level given a price $p$ is $x^*(p) = 6p-3/2$, so the equilibrium price is $p^* = \frac{5}{12}$.
    Then for $p > 1/4$:
    \begin{itemize}
        \item The payment regret at price $p$ is $p(6p-3/2) - 5/12 = 6p^2 - \frac{3}{2}p - 5/12$.
        \item The cost regret at price $p$ is $\frac{1}{12}(6p-3/2)^2 + \frac{1}{4} (6p-3/2) - 1/12 - 1/4 = 3 p^2 - 25/48$.
        \item The relaxed unmet demand is $1- (6p - 3/2) = 5/2 - 6p$.
    \end{itemize}
    And when $p \le 1/4$, the production is 0. So, in this case
    \begin{itemize}
        \item The payment regret is $-5/12$.
        \item The cost regret is $-1/3$.
        \item The relaxed unmet demand is $1$.
    \end{itemize}
\end{enumerate}
Then, the expected total regret at each period $t$ is as follows:
\begin{itemize}
    \item If $p \le 1/4$: expected total regret is $\frac{1}{2}(9p^2/2 - 3p + 1/2 + 1/4) = \frac{9}{4}p^2 - 3p/2 + 3/8$.
    \item If $p > 1/4$: expected total regret is $\frac{1}{2}(9p^2/2 - 3p + 1/2 + 9p^2 - \frac{15}{2}p +\frac{25}{16}) = \frac{27}{4}p^2 - \frac{21}{4}p + 33/32$.
\end{itemize}
From the above obtained relations, we can derive that the expected total regret at any period $t$ is at least $1/96$, which is attained when $p = 7/18$.
It thus follows that, regardless of the pricing strategy adopted by the market operator, the total expected regret is at least $\frac{1}{96} T$.
Hence, either the unmet demand, payment regret, or cost regret are not sublinear, which establishes our claim. 
\section{Regret Measures for Time-Varying Cost Functions} \label{sec:new-regret-defs}
In this section, we will re-define the problem setting and performance metrics with respect to the augmented setting as described in Section~\ref{sec:contextual-bandit}.
For brevity, this section focuses on clarifying the mathematical definitions and we refer the readers to Section~\ref{sec:model} for a complete discussion of the motivations and reasoning behind these definitions.
Recall that, in Section~\ref{sec:contextual-bandit}, we parameterized the suppliers' cost functions with an unknown time-invariant component $\phi_i$ and a time-varying component $\theta_{it}$ that is revealed to the market operator, i.e., 
\[ c_{it}(\cdot) = c(\cdot; \phi_i, \theta_{it}). \]

With this definition in mind, at each period $t$, the suppliers seek to maximize their profits at a given price $p$ through the following optimization problem:
\begin{maxi}|s|[2]                   {x_{it} \geq 0}                               {p x_{it} - c_i(x_{it}; \phi_i, \theta_{it}). \label{eq:supObj-aug}}   {}             {x^*_{i}(p; \phi_i, \theta_{it}) = }                                \end{maxi}
And when the cost functions $c_{it}$'s are convex, we can find the market equilibrium price by solving for the dual variables of the following optimization problem:
\begin{mini!}|s|[2]                   {x_{it} \geq 0, \forall i \in [n]}                               {\sum_{i = 1}^n c_i(x_{it}; \phi_i, \theta_{it}), \label{eq:supObj2-aug}}   {\label{eq:minCost-aug}}             {}                                \addConstraint{\sum_{i = 1}^n x_{it}}{= d_t, \label{eq:demand-con-aug}} 
\end{mini!}
Note that Problems~\eqref{eq:supObj-aug} and~\eqref{eq:supObj2-aug}-\eqref{eq:demand-con-aug} differ from their counterparts in Section~\ref{sec:market-model} (i.e. Problems~\eqref{eq:supObj} and~\eqref{eq:supObj2}-\eqref{eq:demand-con}) only in the parametrization of cost functions $c_{it}$'s.

Like in Section~\ref{sec:perf-measures}, we evaluate the efficacy of an online algorithm for this setting with three regret metrics---unmet demand, payment regret, and cost regret.
Specifically, over the $T$ periods, the market operator sets a sequence of prices $p_t$ according to the online algorithm's policy $\ppi = (\pi_1, \ldots, \pi_T)$, where $p_t = \pi_t(\{ (x_{it'}^*)_{1 = 1}^n, d_{t'}, \theta_{t'} \}_{t'=1}^{t-1}, d_t, \theta_t)$ depends on the past history on the suppliers' production, consumer demands, and contexts. 
The three regret metrics represent the sub-optimality of the policy $\ppi$ relative to the optimal offline algorithm with complete information on the three desirable properties of equilibrium prices as described in Section~\ref{sec:market-model} (i.e. market clearing, minimal supplier cost, and minimal payment).

\paragraph{Unmet Demand:} We evaluate the unmet demand of an online pricing policy $\ppi$ as the sum of the differences between the demand and the total supplier productions corresponding to the pricing policy $\ppi$ at each period $t$. In particular, for an online pricing policy $\ppi$ that sets a sequence of prices $p_1, \ldots, p_T$, the cumulative unmet demand is given by
\begin{align*}
    U_T(\ppi) = \sum_{t = 1}^T \left( d_t - \sum_{i = 1}^n x_{i}^*(p_t; \phi_i, \theta_{it}) \right)_+.
\end{align*}

\paragraph{Cost Regret:} We evaluate the cost regret of an online pricing policy $\ppi$ through the difference between the total supplier production cost corresponding to algorithm $\ppi$ and the minimum total production cost, given complete information on the supplier cost functions. In particular, the cost regret $C_T(\ppi)$ of an algorithm $\ppi$ is given by
\begin{align*}
    C_T(\ppi) = \sum_{t = 1}^T \sum_{i = 1}^n c_{i}(x_{it}^*(p_t; \phi_i, \theta_{it}); \phi_i, \theta_{it}) - c_{i}(x_{it}^*(p^*_t; \phi_i, \theta_{it}); \phi_i, \theta_t),
\end{align*}
where the price $p_t^*$ for each period $t \in [T]$ is the optimal price corresponding to the solution of Problem~\eqref{eq:supObj2}-\eqref{eq:demand-con} given the demand $d_t$ and supplier cost functions $c_{it}$ for all $i \in [n]$.

\paragraph{Payment Regret:} Finally, we evaluate the payment regret of online pricing policy $\ppi$ through the difference between the total payment made to all suppliers corresponding to algorithm $\ppi$ and the minimum total payment, given complete information on the supplier cost functions. In particular, the payment regret $P_T(\ppi)$ of an algorithm $\ppi$ is given by
\begin{align*}
    P_T(\ppi) = \sum_{t = 1}^T \sum_{i = 1}^n p_t x_{i}^*(p_t; \phi_i, \theta_{it})-  p^* x_{it}^*(p^*_t; \phi_i, \theta_{it}).
\end{align*}

We note that, compared to the corresponding definitions in Section~\ref{sec:perf-measures}, we merely explicitly write out the suppliers' cost functions and production levels in terms of the parameterization of an unknown time-invariant component and a known time-varying component as discussed in Section~\ref{sec:contextual-bandit}.

\section{Technical Tools: Freedman Inequality}
The following is a Freedman-type inequality, which is a generalization of Bernstein's inequality to martingale~\cite{freedman1975tail}.
Such inequality has been extensively employed in the bandit literature, e.g., \cite{agarwal2014taming, foster2020beyond}.
\begin{lemma}\label{lem:freedman1}
    Let $\{X_t\}_{t \le T}$ be a martingale difference sequence adapted to a filtration $\{\mathfrak{F}_t\}_{t \le T}$.
        If $|X_t| \le R$ almost surely, then for any $\eta \in (0, 1/R)$,
        \[\sum_{t=1}^{T} X_t \le \eta \sum_{t=1}^{T} \EE[X_t^2 \mid \mathfrak{F}_{t-1}] + \frac{\log(1/\delta)}{\eta}\]
        with probability at least $1-\delta$.
\end{lemma}

By applying this lemma to $X_t - \EE[X_t \mid \mathfrak{F}_{t-1}]$, we have the following
\begin{lemma}\label{lem:freedman2}
    Let $\{X_t\}_{t \le T}$ be a sequence adapted to a filtration $\{\mathfrak{F}_t\}_{t \le T}$. If $0 \le X_t \le R$ almost surely, then
        \[\sum_{t=1}^{T} X_t \le \frac{3}{2} \sum_{t=1}^{T} \EE[X_t \mid \mathfrak{F}_{t-1}] + 4 \log(2/\delta)\]
        and
        \[\sum_{t=1}^{T} \EE[X_t \mid \mathfrak{F}_{t-1}] \le 2 \sum_{t=1}^{T} X_t + 8 R \log(2/\delta)\]
        with probability at least $1-\delta$.
\end{lemma}

\section{Formal Regret Analysis of Algorithm~\ref{alg:time-varying-cost}}
\label{sec:igw-alg-proof}
In this section, we shall formalize the results in Section~\ref{sec:igw-alg-sol} and present a rigorous proof leveraging ideas from~\cite{foster2020beyond}.
Additionally, we will derive the regret guarantees for several additional examples of function classes $\mathcal{F}$ on the suppliers' optimal production with respect to the context $\theta_t$ and the price $p_t$.
First, we quantify the descriptiveness of the function class $\mathcal{F}$ with the following property.
\begin{definition}
A function class $\mathcal{G} : \mathcal{A} \to \mathcal{B}$ is \textit{well-specified} with respect to the ground truth $g^*$ if $g^* \in \mathcal{G}$, and $\mathcal{G}$ is $\varepsilon$-miss-specified with respect to the ground truth $g^*$ if:
\[\exists \bar{g} \in \mathcal{G} \text{, s.t. } \forall a \in \mathcal{A} \text{, we have } \norm{\bar{g}(a) -  g^*(a)}_\mathcal{B} \le \varepsilon. \]
Note that being well-specified is equivalent to being $0$-miss-specified.
\end{definition}

Recall that in Section~\ref{sec:igw-alg-sketch}, we showed that the three desired regret metrics are all upper bounded by the proxy regret (up to some positive constants):
\[ \textsc{Reg}(T) = \sum_{t=1}^T \EE_{p_t \sim \Delta_t} \left[ |x^*(p_t; \theta_t) - d_t|\mid \mathcal{H}_{t-1} \right]. \]
For ease of notation, we use $\lesssim$ as a shorthand notation that the left-hand side is smaller than some fixed constant times the right-hand side.
And we can bound this proxy regret as follows.

\begin{theorem}\label{thm:igw-regret}
    If the function class $\mathcal{F}$ is $\varepsilon$-miss-specified with respect to the suppliers' optimal production $x^*(p_t; \theta_t)$, then with probability $1-\delta$, and for any sequence of contexts $\theta_t$ and demands $d_t$, Algorithm~\ref{alg:time-varying-cost} achieves the following proxy regret bound:
    \[ \textsc{Reg} \lesssim \sqrt{KT \cdot \est(T)} + \varepsilon \sqrt{K} \cdot T + \frac{T}{K} + \sqrt{KT \log(1/\delta)}.\]
\end{theorem}

Noting from Section~\ref{sec:igw-alg-sketch} that the three desired regret metrics are all upper bounded by the proxy regret (up to some positive constants), Theorem~\ref{thm:igw-regret} allows us to formally state the result given by Theorem~\ref{thm:igw-bound-informal} as follows.
\begin{theorem} \label{thm:main-contextual}
    If the function class $\mathcal{F}$ is $\varepsilon$-miss-specified with respect to the suppliers' optimal production $x^*(p_t; \theta_t)$, then with probability $1-\delta$, and for any sequence of contexts $\theta_t$ and demands $d_t$, Algorithm~\ref{alg:time-varying-cost} achieves the following bound on our regret metrics:
    \[ \EE_{p_t \sim \Delta_t, t \in [T]}[U_T(p_1, \dots, p_T)] \lesssim \sqrt{KT \cdot \est(T)} + \varepsilon \sqrt{K} \cdot T + \frac{T}{K} + \sqrt{KT \log(1/\delta)},\]
    and similarly for $\EE_{p_t \sim \Delta_t, t \in [T]}[P_T(p_1, \dots, p_T)]$ and $\EE_{p_t \sim \Delta_t, t \in [T]}[C_T(p_1, \dots, p_T)]$, where $K$ is the number of prices in the uniformly discretized price set.
\end{theorem}

We refer to Appendix~\ref{sec:full-pf-contextual-thm} for a complete proof of Theorem~\ref{thm:igw-regret}. Furthermore, in Appendix~\ref{sec:implications-thm-contextual}, we present explicit regret bounds corresponding to Theorem~\ref{thm:igw-regret} for various function classes $\mathcal{F}$.

\subsection{Regret Bounds for Different Function Classes $\mathcal{F}$} \label{sec:implications-thm-contextual}

In this section, we shall combine Theorem~\ref{thm:main-contextual} with results in the statistical learning literature to establish concrete regret bounds for several examples of function classes $\mathcal{F}$.
To this end, first recall from Section~\ref{sec:igw-alg-sol} that we already obtained an explicit bound on the regret for finite function classes $\mathcal{F}$. Thus, we focus the following discussion on infinite function classes.
To do so, we first quantify the ``size'' of an infinite function class with the notion of \textit{sequential covering}.

\begin{definition}
    Given a real-valued function space $\mathcal{G} : \mathcal{A} \to \RR$ and a sample set $S = \{a_1, \dots, a_n\}$, we say that a finite set of functions $\mathcal{G}'$ is an \textit{$\varepsilon$-sequential cover} of $\mathcal{G}$ with respect to $S$ if 
    \[ \forall \, g \in \mathcal{G}, \exists \, g' \in \mathcal{G}' \text{ s.t. } \left(\frac{1}{n} \sum_{i=1}^n (g(a_i) - g'(a_i))^2\right)^{1/2} < \varepsilon. \]
    Then, the \textit{$\varepsilon$-sequential covering number} of $\mathcal{G}$ is the size of the smallest $\varepsilon$-sequential cover with respect to the sample set $S$, 
    \[ \mathcal{N}_2(\mathcal{G}, \varepsilon, S) = \min \{|\mathcal{G}'| : \mathcal{G}' \text{ is a $\varepsilon$-sequential cover of $\mathcal{G}$ with respect to $S$}\}. \]
    Finally, denote $\mathcal{N}_2(\mathcal{G}, \varepsilon) = \sup_{S \text{ finite}} \mathcal{N}_2(\mathcal{G}, \varepsilon, S)$.
\end{definition}

As shown in~\cite{rakhlin2014online}, the prediction accuracy can be expressed in terms of the sequential covering number of the function class $\mathcal{F}$.
\begin{theorem}
    (see \cite{rakhlin2014online})
    There exist online regression oracles achieving the following bounds:
    \begin{itemize}
        \item If $\mathcal{F}$ is finite, then $\est(T) \le \log |\mathcal{F}|$.
        \item If $\mathcal{F}$ is parametric in the sense that $\mathcal{N}_2(\mathcal{F}, \varepsilon) \in O(\varepsilon^{-m})$, then $\est(T) \lesssim m \cdot \log(T)$.
        \item If $\mathcal{F}$ is non-parametric in the sense that $\log \mathcal{N}_2(\mathcal{F}, \varepsilon) \in O(\varepsilon^{-m})$, then $\est(T) \lesssim T^{1 - 2/(2+m)}$ if $m \in (0, 2)$ and $\est(T) \lesssim T^{1 - 1/m}$ if $m \ge 2$.
    \end{itemize}
\end{theorem}
In many cases, we can easily construct efficient algorithms that match or nearly match these bounds.
For example, for finite $\mathcal{F}$, we can achieve the bound $\est(T) \le \log |\mathcal{F}|$ with the classical exponential weights update algorithm~\citep{vovk1995game}.
And when $\mathcal{F}$ is a linear class in the sense that
\[\mathcal{F} =\{(p, \theta) \mapsto \langle \phi, \sigma(p, \theta) \rangle : \phi \in B^m_2\},\]
where $\sigma$ is a fixed feature map, then the Vovk-Azoury-Warmuth forecaster achieves $\est(T) \lesssim m \cdot \log(T)$~\citep{vovk1997competitive,azoury2001relative}.
For general function class $\mathcal{F}$, we can nearly achieve the preceding bounds by performing the exponential weights update algorithm on a sequential cover of $\mathcal{F}$ (see e.g.,~\cite{vovk2006metric}).

With these results in mind, we can derive the exact regret bounds for various instances of function class $\mathcal{F}$.
In particular, if the suppliers' optimal production function $x^*$ is contained in one of the function classes listed below, we can compute the regret bounds as follows:
\begin{itemize}
    \item If $\mathcal{F}$ is finite, we have $\est(T) = \log |\mathcal{F}|$ and a choice of $K = (\frac{T}{\log |\mathcal{F}|})^{1/3}$ results in the regret bound $\textsc{Reg}(T) \lesssim T^{2/3} \left(\sqrt[3]{\log |\mathcal{F}|} + \sqrt{\log(1/\delta)}\right)$ with probability $1-\delta$.
    \item If the cost functions are quadratic functions with time-varying coefficients in the sense that
    \[c(x; \phi, \theta_t) = \frac{1}{2\langle \phi, \sigma(\theta_t)\rangle} x^2, \; \phi \in B^m_2,\]
    for some fixed feature map $\sigma$, then the suppliers' optimal production can be expressed as
    \[x^*(p; \phi, \theta_t) = \langle \phi, p \cdot \sigma(\theta_t) \rangle.\]
    So, in this case, $\mathcal{F}$ is a linear function class and we have $\est(T) \lesssim m \cdot \log(T)$.
    If $K = (\frac{T}{d\log T})^{1/3}$, then we get the regret bound $\textsc{Reg}(T) \lesssim T^{2/3} \left(\sqrt[3]{m \cdot \log T} + \sqrt{\log(1/\delta)}\right)$ with probability $1-\delta$.
    \item Consider an Euclidean context $\theta_t \in \RR^m$ encapsulates similarity information on the suppliers' behavior such that $\mathcal{F}$ is the set of bounded Lipschitz functions over $(p, \theta) \in \RR^{m+1}$.
    From this well-specified function class $\mathcal{F}$, we find a subset of $\mathcal{F}$ with sequential covering and apply the miss-specified version of Theorem~\ref{thm:igw-regret} on this subset.
    We can explicitly construct an $\varepsilon$-covering so that $\log \mathcal{N}_2(\mathcal{F}, \varepsilon) \lesssim \varepsilon^{-m-1} $ (see Examples 5.10 and 5.11 in~\cite{wainwright2019high}).
    If we run the exponential weights update algorithm over this covering, we have $\est(T) = \varepsilon^{-m-1}$.
    Since, by construction, the covering is $\varepsilon$-miss-specified with respect to the optimal suppliers' production, we have
    \[\textsc{Reg}(T) \lesssim \sqrt{KT \varepsilon^{-m-1}} + \varepsilon \sqrt{K} \cdot T + \frac{T}{K} + \sqrt{KT \log(1/\delta)}.\]
    If we pick $K = \varepsilon^{-2/3}$ and $\varepsilon = T^{-1/(m+2)}$, then we get that
    \[\textsc{Reg}(T) \lesssim T^{(3m+4)/(3m+6)} + T^{1/2 + 1/(3m+6)} \sqrt{\log(1/\delta)}\]
    with probability $1-\delta$.
    \item If $\mathcal{F}$ is a neural network whose weight matrices' spectral norms are at most 1, then it is known that $\log \mathcal{N}_2(\mathcal{F}, \varepsilon) \lesssim \varepsilon^{-2}$~\citep{bartlett2017spectrally}. 
    Then, we have $\est(T) \in O(T^{1/2})$, and a choice of $K = T^{1/6}$ gives the regret bound $\textsc{Reg} \lesssim T^{5/6} + T^{7/12} \sqrt{\log(1/\delta)}$ with probability $1-\delta$.
\end{itemize}

We can analogously work out the regret bounds if each of the examples is miss-specified.

\subsection{Proof of Theorem~\ref{thm:igw-regret}} \label{sec:full-pf-contextual-thm}

Recall that $x^*$ is the suppliers' optimal production, and since $\mathcal{F}$ is $\varepsilon$-miss-specified with respect to $x^*$, there exists a function $\bar{x} \in \mathcal{F}$ so that
\[\forall (p, \theta), \text{we have } |\bar{x}(p; \theta) -  x^*(p; \theta)| \le \varepsilon.\]
From Lemma~\ref{lem:ProductionLipschitz}, we know that $x^*$ is Lipschitz in price, so we know that for any context $\theta$ and price $p$, there exists $\bar{p}$ from the list of Algorithm~\ref{alg:time-varying-cost}'s choices $\{p_i\}_{i=1}^K$ such that $|x^*(\theta, p) - x^*(\theta, \bar{p})| < O(1/K)$.
Note that for the market-clearing price $p^*_t$, we have $|x^*(p^*_t; \theta_t) - d_t| = 0$.
Therefore, from the triangle inequality we have
\begin{align*}
    \textsc{Reg}(T) 
    &= \sum_{t=1}^T \EE_{p_t \sim \Delta_t} \left[ |x^*(p_t; \theta_t) - d_t|\mid \mathcal{H}_{t-1} \right] - |x^*(p^*_t; \theta_t) - d_t| \\
    &\le \sum_{t=1}^T \EE_{p_t \sim \Delta_t} \left[ |\bar{x}(p_t; \theta_t) - d_t| \mid \mathcal{H}_{t-1} \right] - |\bar{x}( \bar{p}_t; \theta_t) - d_t| + 2\varepsilon T + O(T/K)
\end{align*}
Now we attempt to upper bound the quantity
\[ \preg(T) := \sum_{t=1}^T \EE_{p_t \sim \Delta_t} \left[ |\bar{x}(p_t; \theta_t) - d_t| \mid \mathcal{H}_{t-1} \right] - |\bar{x}(\bar{p}_t; \theta_t) - d_t|. \]
Recall that $\hat{f}_t$ is the online regression oracle's output at time $t$.
For simplicity, denote
\begin{align*}
    \bar{g}_t(p) &= |\bar{x}(p; \theta_t) - d_t|, \\
    \hat{g}_t(p) &= |\hat{f}_t(p; \theta_t) - d_t|.
\end{align*}
Let $\hat{p}_t = \argmin \hat{g}_t(\cdot)$, then we have
\begin{align*}
    & \EE_{p_t \sim \Delta_t} [\bar{g}_t(p_t) \mid \mathcal{H}_{t-1} ] - \bar{g}_t(\bar{p}_t) \\
    ={}& \underbrace{\EE_{p_t \sim \Delta_t} [\hat{g}_t(p_t) - \hat{g}_t(\hat{p}_t) \mid \mathcal{H}_{t-1} ]}_{\textcircled{1}} + \underbrace{\EE_{p_t \sim \Delta_t} [\bar{g}_t(p_t) - \hat{g}_t(p_t) \mid \mathcal{H}_{t-1} ]}_{\textcircled{2}} + \underbrace{(\hat{g}_t(\hat{p}_t) - \bar{g}_t(\bar{p}_t))}_{\textcircled{3}}
\end{align*}
Plugging in the chosen distribution $\Delta_t$, we get that the first term is
\[\textcircled{1} = \sum_{i=1}^K \frac{\hat{g}_t(p_i) - \hat{g_t}(\hat{p_t})}{\lambda + 2\gamma (\hat{g}_t(p_i) - \hat{g_t}(\hat{p_t}))} \le \frac{K-1}{2\gamma}. \]
By convexity and then AM-GM, the second term can be bounded as 
\[\textcircled{2} \le \sqrt{\EE_{p_t \sim \Delta_t} [(\bar{g}_t(p_t) - \hat{g}_t(p_t))^2 \mid \mathcal{H}_{t-1}]} \le \frac{1}{2\gamma} + \frac{\gamma}{2} \EE_{p_t \sim \Delta_t} [(\bar{g}_t(p_t) - \hat{g}_t(p_t))^2 \mid \mathcal{H}_{t-1}].\]
And the third term can be rewritten as
\begin{align*}
    \textcircled{3}
    &= \hat{g}_t(\bar{p}_t) - \bar{g}(\bar{p}_t) - (\hat{g}_t(\bar{p}_t) - \hat{g}_t(\hat{p}_t)) \\
    &\le \frac{\gamma \Delta_t(\bar{p}_t)}{2}(\hat{g}_t(\bar{p}_t) - \bar{g}(\bar{p}_t))^2 + \frac{1}{2\gamma \Delta_t(\bar{p}_t)} - (\hat{g}_t(\bar{p}_t) - \hat{g}_t(\hat{p}_t)) \\
    &\le \frac{\gamma}{2} \EE_{p_t \sim \Delta_t}\left[ (\hat{g}_t(\bar{p}_t) - \bar{g}(\bar{p}_t))^2 \mid \mathcal{H}_{t-1} \right] + \frac{1}{2\gamma \Delta_t(\bar{p}_t)} - (\hat{g}_t(\bar{p}_t) - \hat{g}_t(\hat{p}_t)) \\
    &= \frac{\gamma}{2} \EE_{p_t \sim \Delta_t}\left[ (\hat{g}_t(\bar{p}_t) - \bar{g}(\bar{p}_t))^2 \mid \mathcal{H}_{t-1} \right] + \frac{\lambda + 2\gamma (\hat{g}_t(\bar{p}_t) - \hat{g}_t(\hat{p}_t))}{2\gamma } - (\hat{g}_t(\bar{p}_t) - \hat{g}_t(\hat{p}_t)) \\
    &= \frac{\gamma}{2} \EE_{p_t \sim \Delta_t}\left[ (\hat{g}_t(\bar{p}_t) - \bar{g}(\bar{p}_t))^2 \mid \mathcal{H}_{t-1} \right] + \frac{\lambda}{2\gamma} \\
    &\le \frac{\gamma}{2} \EE_{p_t \sim \Delta_t}\left[ (\hat{g}_t(\bar{p}_t) - \bar{g}(\bar{p}_t))^2 \mid \mathcal{H}_{t-1} \right] + \frac{K}{2\gamma}
\end{align*}
Now, summing these three terms yields that
\[ \EE_{p_t \sim \Delta_t} [\bar{g}_t(p_t) \mid \mathcal{H}_{t-1} ] - \bar{g}_t(\bar{p}_t) \le \frac{K}{\gamma} + \gamma \cdot \EE_{p_t \sim \Delta_t} [(\bar{g}_t(p_t) - \hat{g}_t(p_t))^2 \mid \mathcal{H}_{t-1}]. \]
After summing over $t =1, \dots, T$, we have
\[ \preg(T) = \sum_{t=1}^T \EE_{p_t \sim \Delta_t} [\bar{g}_t(p_t) \mid \mathcal{H}_{t-1} ] - \bar{g}_t(\bar{p}_t) \le \frac{KT}{\gamma} + \gamma \sum_{t=1}^T \EE_{p_t \sim \Delta_t} [(\bar{g}_t(p_t) - \hat{g}_t(p_t))^2 \mid \mathcal{H}_{t-1}]. \]
Next, we upper bound the RHS with some case work.
\begin{itemize}
    \item  $\bar{x}(p_t; \theta_t) \ge d_t \ge \hat{f}_t(p_t; \theta_t)$:
    \begin{align*}
        (|\bar{x}(p_t; \theta_t) - d_t| - |\hat{f}_t(p_t; \theta_t) - d_t|)^2 
        &\le (|\bar{x}(p_t; \theta_t) - d_t| + |\hat{f}_t(p_t; \theta_t) - d_t|)^2 \\
        &= (\bar{x}(p_t; \theta_t) - d_t + d_t - \hat{f}_t(p_t; \theta_t))^2 \\ 
        &= (\bar{x}(p_t; \theta_t) - \hat{f}_t(p_t; \theta_t))^2
    \end{align*}
    \item $\hat{f}_t(p_t; \theta_t) \ge d_t \ge \bar{x}(p_t; \theta_t)$, similar to the previous case, we have
    \[(|\bar{x}(p_t; \theta_t) - d_t| - |\hat{f}_t(p_t; \theta_t) - d_t|)^2 \le (\bar{x}(p_t; \theta_t) - \hat{f}_t(p_t; \theta_t))^2 \]
    \item $\hat{f}_t(p_t; \theta_t), \bar{x}(p_t; \theta_t) \ge d_t$ or $\bar{x}(p_t; \theta_t), \hat{f}_t(p_t; \theta_t) \le d_t$:
    \[(|\bar{x}(p_t; \theta_t) - d_t| - |\hat{f}_t(p_t; \theta_t) - d_t|)^2 = (\bar{x}(p_t; \theta_t) - d_t + d_t - \hat{f}_t(p_t; \theta_t))^2 = (\bar{x}(p_t; \theta_t) - \hat{f}_t(p_t; \theta_t))^2 \]
\end{itemize}
Therefore,
\begin{equation}\label{equ:reg-to-est}
    \preg(T) = \sum_{t=1}^T \EE_{p_t \sim \Delta_t} [\bar{g}_t(p_t) \mid \mathcal{H}_{t-1} ] - \bar{g}_t(\bar{p}_t) \le \frac{KT}{\gamma} + \gamma \sum_{t=1}^T \EE_{p_t \sim \Delta_t} [(\bar{x}(p_t; \theta_t) - \hat{f}_t(p_t; \theta_t))^2 \mid \mathcal{H}_{t-1}].
\end{equation}
We finish the proof by claiming that the final term can be bounded by the online regression oracle's prediction error guarantee with high probability.
Since $\bar{x}$ and $\hat{f}_t$ are bounded, we can apply Lemma~\ref{lem:freedman2} to get
\begin{equation}\label{equ:freedman1}
    \sum_{t=1}^T \EE_{p_t \sim \Delta_t} [(\bar{x}(p_t; \theta_t) - \hat{f}_t(p_t; \theta_t))^2 \mid \mathcal{H}_{t-1}] \le 2 \sum_{t=1}^T (\bar{x}(p_t) - \hat{f}_t(p_t))^2 + O(\log(1/\delta))
\end{equation}
Let $x_t$ be our observation on suppliers' production at time $t$.
We expand the RHS as follows:
\begin{align*}
    (\bar{x}(p_t; \theta_t) - \hat{f}_t(p_t; \theta_t))^2
    &= \bar{x}(p_t; \theta_t)^2 - 2\bar{x}(p_t; \theta_t)\hat{f}_t(p_t; \theta_t) + \hat{f}_t(p_t; \theta_t)^2 \\
    &= 2 \bar{x}(p_t; \theta_t)^2 - 2\bar{x}(p_t; \theta_t)\hat{f}_t(p_t; \theta_t) + 2 x_t (\hat{f}_t(p_t; \theta_t) - \bar{x}(p_t; \theta_t)) \\
    &\hspace{10em}+ \hat{f}_t(p_t; \theta_t)^2 - 2 x_t (\hat{f}_t(p_t; \theta_t) - \bar{x}(p_t; \theta_t)) - \bar{x}(p_t; \theta_t)^2 \\
    &= 2(x_t - \bar{x}(p_t; \theta_t))(\hat{f}_t(p_t; \theta_t)  - \bar{x}(p_t; \theta_t)) + (\hat{f}_t(p_t; \theta_t) - x_t)^2 - (\bar{x}(p_t; \theta_t) - x_t)^2 
\end{align*}
Therefore,
\begin{equation}\label{equ:diff-expand}
\begin{aligned}
    & \sum_{t=1}^T (\bar{x}(p_t; \theta_t) - \hat{f}_t(p_t; \theta_t))^2 \\
    ={}& \sum_{t=1}^T (\hat{f}_t(p_t; \theta_t) - x_t)^2 - \sum_{t=1}^T (\bar{x}(p_t; \theta_t) - x_t)^2 + \sum_{t=1}^T 2(x_t - \bar{x}(p_t; \theta_t))(\hat{f}_t(p_t; \theta_t)  - \bar{x}(p_t; \theta_t)) \\
    \le{}& \sum_{t=1}^T (\hat{f}_t(p_t; \theta_t) - x_t)^2 - \inf_{f \in \mathcal{F}}\sum_{t=1}^T (f(p_t; \theta_t) - x_t)^2 + \sum_{t=1}^T 2(x_t - \bar{x}(p_t; \theta_t))(\hat{f}_t(p_t; \theta_t)  - \bar{x}(p_t; \theta_t)) \\
    ={}& \est(T) + \sum_{t=1}^T 2(x_t - \bar{x}(p_t; \theta_t))(\hat{f}_t(p_t; \theta_t)  - \bar{x}(p_t; \theta_t))
\end{aligned}
\end{equation}
The last term can be expanded as
\begin{align*}
    &\sum_{t=1}^T (x_t - \bar{x}(p_t; \theta_t))(\hat{f}_t(p_t; \theta_t)  - \bar{x}(p_t; \theta_t)) \\
    ={}& \sum_{t=1}^T (x_t - x^*(p_t; \theta_t) + x^*(p_t; \theta_t) - \bar{x}(p_t; \theta_t))(\hat{f}_t(p_t; \theta_t)  - \bar{x}(p_t; \theta_t)) \\
    \le{}& \sum_{t=1}^T (x_t - x^*(p_t; \theta_t))(\hat{f}_t(p_t; \theta_t)  - \bar{x}(p_t; \theta_t)) + \sum_{t=1}^T \varepsilon\cdot |\hat{f}_t(p_t; \theta_t)  - \bar{x}(p_t; \theta_t)| \\
    \le{}& \sum_{t=1}^T (x_t - x^*(p_t; \theta_t))(\hat{f}_t(p_t; \theta_t)  - \bar{x}(p_t; \theta_t)) + \sum_{t=1}^T 6 \varepsilon^2 + \frac{1}{24}(\hat{f}_t(p_t; \theta_t)  - \bar{x}(p_t; \theta_t))^2
\end{align*}
Since $\EE[x_t - x^*(p_t; \theta_t) \mid \mathcal{H}_{t-1}] = 0$ and both $x_t - x^*(p_t; \theta_t)$ and $\hat{f}_t(p_t; \theta_t)  - \bar{x}(p_t; \theta_t)$ are bounded above by $2\bar{d}$, we can apply Lemma~\ref{lem:freedman1} to the first sum with $\eta = 1/(64\bar{d}^2)$.
And for the second sum, we can apply Lemma~\ref{lem:freedman2}.
Then, with probability at least $1-\delta$, we have
\begin{equation}\label{equ:freedman2}
\begin{aligned}
    &\sum_{t=1}^T (x_t - \bar{x}(p_t; \theta_t))(\hat{f}_t(p_t; \theta_t)  - \bar{x}(p_t; \theta_t)) \\
    \le{}& \sum_{t=1}^T \left(\frac{1}{64 \bar{d}^2} \EE[(x_t - f^*(p_t; \theta_t))^2(\hat{f}_t(p_t; \theta_t)  - \bar{x}(p_t; \theta_t))^2 \mid \mathcal{H}_{t-1}] + \frac{1}{16}\EE[(\hat{f}_t(p_t; \theta_t)  - \bar{x}(p_t; \theta_t))^2 \mid \mathcal{H}_{t-1}]\right) \\
    &\hspace{26em}+ O(\varepsilon^2 T + \log(1/\delta)) \\
    \le{}& \frac{1}{8} \sum_{t=1}^T  \EE[(\hat{f}_t(p_t; \theta_t)  - \bar{x}(p_t; \theta_t))^2 \mid \mathcal{H}_{t-1}] + O(\varepsilon^2 T + \log(1/\delta))
\end{aligned}
\end{equation}
Putting \eqref{equ:freedman1}, \eqref{equ:diff-expand} and \eqref{equ:freedman2} together gives us
\begin{align*}
    & \sum_{t=1}^T \EE_{p_t \sim \Delta_t} [(\bar{x}(p_t; \theta_t) - \hat{f}_t(p_t; \theta_t))^2 \mid \mathcal{H}_{t-1}] \\
    \le{}& 2 \est(T) + \frac{1}{2} \sum_{t=1}^T \EE_{p_t \sim \Delta_t} [(\bar{x}(p_t; \theta_t) - \hat{f}_t(p_t; \theta_t))^2 \mid \mathcal{H}_{t-1}] + O(\varepsilon^2 T + \log(1/\delta))
\end{align*}
Therefore, 
\begin{equation}\label{equ:est-to-oracle}
    \sum_{t=1}^T \EE_{p_t \sim \Delta_t} [(\bar{x}(p_t; \theta_t) - \hat{f}_t(p_t; \theta_t))^2 \mid \mathcal{H}_{t-1}] \le 4 \est(T) + O(\varepsilon^2 T + \log(1/\delta))
\end{equation}
Combining \eqref{equ:reg-to-est} and \eqref{equ:est-to-oracle}, and taking
\[\gamma = \sqrt{\frac{KT}{\est(T) + \varepsilon^2 T + \log(1/\delta)}},\]
we conclude the following regret bound with probability $1-\delta$:
\begin{align*}
    \preg(T)
    &\lesssim \frac{KT}{\gamma} + \gamma \cdot \left(\est(T) + \varepsilon^2 T + \log(1/\delta)\right) \\
    &\lesssim \sqrt{KT(\est(T) + \varepsilon^2 T + \log(1/\delta))} \\
    &\lesssim \sqrt{KT \cdot\est(T)} + \varepsilon\sqrt{K} \cdot T + \sqrt{KT \log(1/\delta)}
\end{align*}
Finally, we have
\begin{align*}
    \textsc{Reg}(T) 
    &\lesssim \sqrt{KT \cdot\est(T)} + \varepsilon\sqrt{K} \cdot T + \sqrt{KT \log(1/\delta)} + \varepsilon T + \frac{T}{K} \\
    &\lesssim \sqrt{KT \cdot\est(T)} + \frac{T}{K} + \varepsilon\sqrt{K} \cdot T + \sqrt{KT \log(1/\delta)},
\end{align*}
with probability $1-\delta$.
so we are done.

\end{document}